\documentclass[11pt]{article}

\usepackage{amsmath} 
\usepackage{amsthm} 
\usepackage{amssymb}	
\usepackage{graphicx} 
\usepackage{multicol} 
\usepackage{multirow}
\usepackage{xcolor}
\usepackage{bm}
\usepackage{bbm}
\usepackage[letterpaper,margin=1in,bottom=1in]{geometry}
\usepackage{authblk}

\usepackage[utf8]{inputenc}
\usepackage[english]{babel}

\usepackage{mathtools}

\newtheorem{theorem}{Theorem}[section]

\newtheorem{lemma}[theorem]{Lemma}
\newtheorem{corollary}[theorem]{Corollary}
\newtheorem{proposition}[theorem]{Proposition}

\newtheorem{definition}{Definition}[section]

\DeclarePairedDelimiter\ket{\lvert}{\rangle}
\DeclarePairedDelimiter\bra{\langle}{\rvert}

\newcommand{\proj}[1]{\ket{#1}\!\bra{#1}}

\DeclarePairedDelimiter\norm{\lVert}{\rVert}

\usepackage{algorithm}
\usepackage{algpseudocode}

\usepackage{tabularx}
\usepackage{booktabs}
\usepackage{threeparttable}
\usepackage{adjustbox} 

\usepackage[
    colorlinks=true,
    linkcolor=blue!55!black,
    urlcolor=blue!55!black,
    citecolor=blue!55!black
]{hyperref}
\usepackage[capitalize,noabbrev]{cleveref}
\makeatletter
\providecommand{\theHALG@line}{\thealgorithm.\arabic{ALG@line}}
\makeatother

\newcommand{\cO}{\mathcal O}
\newcommand{\Tr}{\operatorname{Tr}}
\newcommand{\Dtr}{\mathrm{d}_{\mathrm{tr}}}
\newcommand{\transduce}[1]{\stackrel{#1}{\rightsquigarrow}}
\newcommand{\DownTransduce}{%
  \scalebox{0.7}{\rotatebox[origin=c]{270}{$\rightsquigarrow$}}}

\title{Optimal Query Complexity for Ground-State Preparation}

\author[1]{Boyang Chen\thanks{\texttt{by-chen24@mails.tsinghua.edu.cn}}}
\author[2,3]{Minbo Gao\thanks{\texttt{gmb17@tsinghua.org.cn}}}
\author[4,5]{Xinzhao Wang\thanks{\texttt{xinzhaowang3@gmail.com}}}
\author[4,5]{Shuo Zhou\thanks{\texttt{antientropy@pku.edu.cn}}}

\affil[1]{Department of Computer Science and Technology,
Tsinghua University, Beijing, China}

\affil[2]{Institute of Software,
Chinese Academy of Sciences, Beijing, China}

\affil[3]{University of Chinese Academy of Sciences,
Beijing, China}

\affil[4]{Center on Frontiers of Computing Studies,
Peking University, Beijing, China}

\affil[5]{School of Computer Science,
Peking University, Beijing, China}
\date{}

\begin{document}
\maketitle

\begin{abstract}
We determine the optimal query complexity of ground-state preparation to
trace-distance error $\varepsilon$ when an energy threshold in the spectral
gap is known.  Let $U_H$ be an $\alpha$-block-encoding of a Hamiltonian
with unique ground state $\ket{\psi_0}$, and suppose
$|\langle\psi_0|U_I|0\rangle|\ge\gamma$ for a state-preparation oracle $U_I$.
The threshold lies at least $\Delta/2$ above the ground-state energy and at
least $\Delta/2$ below every excited-state energy.  We give two algorithms
that prepare a state within trace distance $\varepsilon$ of the ground state.
One uses $\cO((\alpha/\Delta)(\gamma^{-1}+\log(1/\varepsilon)))$ calls to
$U_H$ in expectation; the other uses
$\cO((\alpha/(\gamma\Delta))\log(1/\varepsilon))$ calls to $U_H$ in the worst
case.  We prove a lower bound matching the expected query count; the
corresponding worst-case lower bound follows from Somma and de
Wolf~\cite{SdW26}.  The respective bounds on
calls to $U_I$ are $\cO(1/\gamma)$ in expectation and
$\cO(\gamma^{-1}\log(1/\varepsilon))$ in the worst case.  On $(N+1)$-dimensional systems,
these $U_I$ bounds are also optimal when the expected or worst-case count of
$U_H$ calls, respectively, is $o((\alpha/\Delta)\sqrt N)$.  Both algorithms
use a constant-accuracy spectral filter to construct a purifier, which we
then sequentially compose during amplitude amplification to prepare a state
with constant overlap with the ground state.  The expected-query algorithm
repeats the preparation followed by one high-accuracy spectral filter until
success.  The worst-case algorithm uses filters of increasing accuracy and
limits the total number of queries.
\end{abstract}

\newpage
\tableofcontents

\newpage

\section{Introduction}
\label{sec:introduction}

Ground-state preparation is a fundamental task in quantum algorithms for
many-body physics and quantum chemistry, enabling the study of ground-state
properties~\cite{PW09,McArdle20}.  Preparing such states is computationally
hard in general: an efficient method for arbitrary local Hamiltonians,
followed by energy estimation, would solve the QMA-complete Local Hamiltonian
problem~\cite{KKR06}.  Algorithms with provable efficiency guarantees
therefore typically assume additional information about the initial state
and the low-energy spectrum.

We consider a Hamiltonian $H$ with a unique ground state $\ket{\psi_0}$.  The
Hamiltonian is accessed through an $\alpha$-block-encoding $U_H$, a unitary whose
all-zero ancilla block is $H/\alpha$~\cite{GSLW19}.  A second oracle $U_I$
prepares a state $\ket\phi$ satisfying
$|\langle\psi_0|\phi\rangle|\ge\gamma$.  We assume a known threshold $\mu$
such that the ground-state energy lies at least $\Delta/2$ below $\mu$ and
every excited-state energy lies at least $\Delta/2$ above $\mu$.  The goal is to prepare
$\ket{\psi_0}$ to trace-distance error at most $\varepsilon$.  Let $q_H,q_I$
count calls to $U_H,U_H^\dagger$ and $U_I,U_I^\dagger$, respectively, and let
$Q_H,Q_I$ be worst-case bounds that hold on every run.

Lin and Tong~\cite[Section~3 and Theorem~6]{LT20} use quantum singular value
transformation (QSVT)~\cite{GSLW19} to construct a spectral filter that
suppresses the excited-state components of $\ket\phi$.  The filter succeeds
with probability $\Omega(\gamma^2)$.  Amplitude amplification
raises this probability to a constant in $\cO(1/\gamma)$ iterations~\cite{BHMT02}.
Each amplitude-amplification iteration uses a reflection about $\ket\phi$, implemented with $U_I$ and
$U_I^\dagger$, and a reflection about the ground state.  Approximating the
latter to error $\cO(\gamma\varepsilon)$ per use keeps the total amplification
error within $\cO(\varepsilon)$.  The spectral filter can be used to
approximate this ground-state reflection using
$\cO((\alpha/\Delta)\log(1/(\gamma\varepsilon)))$ calls to $U_H$ and
$U_H^\dagger$.  Thus, their algorithm succeeds with constant
probability, and its worst-case query counts satisfy
\begin{equation}
Q_H=\cO\!\left(
\frac{\alpha}{\gamma\Delta}
\log\frac1{\gamma\varepsilon}
\right),
\qquad
Q_I=\cO\!\left(\frac1\gamma\right).
\label{eq:intro-baseline}
\end{equation}
Repeating independent runs of this algorithm until success produces a state
within trace distance $\varepsilon$ of $\proj{\psi_0}$, with $\mathbb E[q_H]$
and $\mathbb E[q_I]$ obeying the corresponding bounds
in~\eqref{eq:intro-baseline}.

In this paper, we give an algorithm whose expected number of calls to
$U_H$ satisfies
\begin{equation}
\mathbb E[q_H]=\cO\!\left(
\frac{\alpha}{\gamma\Delta}
+\frac{\alpha}{\Delta}\log\frac1\varepsilon
\right).
\label{eq:intro-optimal}
\end{equation}
It also has $\mathbb E[q_I]=\cO(1/\gamma)$.  We prove a matching lower bound on
$\mathbb E[q_H]$.

We also give an algorithm whose output state is within trace distance
$\varepsilon$ of $\proj{\psi_0}$ and whose query counts are bounded on every
run as follows:
\begin{equation}
Q_H=\cO\!\left(
\frac{\alpha}{\gamma\Delta}\log\frac1\varepsilon
\right),
\qquad
Q_I=\cO\!\left(\frac1\gamma\log\frac1\varepsilon\right).
\label{eq:intro-unconditional}
\end{equation}
The worst-case $Q_H$ bound matches the lower bound of Somma and de
Wolf~\cite{SdW26}.  We also prove matching lower bounds on calls to $U_I$
when $\mathbb E[q_H]$ in the expected-query setting or $Q_H$ in the
worst-case setting is $o((\alpha/\Delta)\sqrt N)$ on
$(N+1)$-dimensional systems:
\[
\mathbb E[q_I]=\Omega(1/\gamma),
\qquad
Q_I=\Omega\!\left(\frac1\gamma\log\frac1\varepsilon\right).
\]
A restriction on calls to $U_H$ is needed for these $U_I$ lower bounds:
at constant error,
$\cO((\alpha/\Delta)\sqrt N)$ calls to $U_H$ suffice without $U_I$;
see Section~\ref{subsec:initial-query-lower}.

To obtain the expected bound~\eqref{eq:intro-optimal}, we make two changes to
the filtering approach behind~\eqref{eq:intro-baseline}.  First, amplitude
amplification is used only to produce a state $\rho_c$ satisfying
$\bra{\psi_0}\rho_c\ket{\psi_0}=\Omega(1)$; a single high-accuracy spectral
filter then yields a state within trace distance $\varepsilon$ of
$\ket{\psi_0}$ on success.  If the ground-state reflections in amplitude
amplification are implemented by QSVT circuits, the two-stage algorithm's
worst-case query count satisfies
\[
Q_H=\cO\!\left(
\frac{\alpha}{\gamma\Delta}\log\frac1\gamma
+\frac{\alpha}{\Delta}\log\frac1\varepsilon
\right).
\]
Second, we combine the space-efficient purifier of Belovs and
Jeffery~\cite{BJ26} with the transducer composition results of Belovs,
Jeffery, and Yolcu~\cite{BJY24} to remove the remaining $\log(1/\gamma)$.
This allows the $\cO(1/\gamma)$ ground-state reflections to be composed
without constructing a separate $\cO(\gamma)$-accurate circuit for each
reflection.  The resulting
two-stage algorithm succeeds with constant probability per run, so repeating
it until success gives
Eq.~\eqref{eq:intro-optimal}.

Our use of transducers is closely related to the ground-state energy estimation algorithm of Jeffery and Witteveen~\cite{JW26}, which removes the same $\log(1/\gamma)$ overhead using queries to a
unitary whose eigenphases encode the Hamiltonian spectrum. Their algorithm is organized around decision problems, composing transducers for decision versions of phase estimation and amplitude amplification to estimate the ground-state energy. In contrast, we use transducers within amplitude amplification to prepare a state with constant ground-state overlap, and then apply high-accuracy spectral filtering to reach error $\varepsilon$.

For the fixed query bound~\eqref{eq:intro-unconditional}, one could instead
run the two-stage algorithm $\cO(\log(1/\varepsilon))$ times, with
target error $\varepsilon/2$, and
return an arbitrary state if every run fails.  Repeating the entire procedure
would multiply both terms of its query cost:
\begin{equation}
Q_H=\cO\!\left[
\left(\frac{\alpha}{\gamma\Delta}
+\frac{\alpha}{\Delta}\log\frac1\varepsilon\right)
\log\frac1\varepsilon
\right].
\label{eq:intro-naive-fixed}
\end{equation}
To avoid paying for the most accurate filter on every attempt, our algorithm
applies filters of increasing accuracy and starts a new attempt when one rejects.
Rejected attempts rarely reach the more expensive filters.  A query cutoff
then gives~\eqref{eq:intro-unconditional} without the extra
$\log(1/\varepsilon)$ multiplying the filtering cost
in~\eqref{eq:intro-naive-fixed}.

Table~\ref{tab:query-comparison} compares these bounds with those of
Lin and Tong~\cite[Theorem~6]{LT20} under the same oracle model.
\begin{table}[htbp]
\centering
\begin{threeparttable}
\vspace{-8pt}
\caption{Query bounds for ground-state preparation to trace-distance error
$\varepsilon$.}
\label{tab:query-comparison}
\small
\setlength{\tabcolsep}{5pt}
\renewcommand{\arraystretch}{1.25}
\begin{tabular}{@{}lcll@{}}
\toprule
Algorithm & Query measure & $U_H,U_H^\dagger$ calls & $U_I,U_I^\dagger$ calls \\
\midrule
Lin--Tong~\cite[Theorem~6]{LT20} & Expected
  & $\cO\!\left(\dfrac{\alpha}{\gamma\Delta}\log\dfrac1{\gamma\varepsilon}\right)$
  & $\cO(1/\gamma)$ \\
This work & Expected
  & $\cO\!\left(\dfrac{\alpha}{\gamma\Delta}
      +\dfrac{\alpha}{\Delta}\log\dfrac1\varepsilon\right)$
  & $\cO(1/\gamma)$ \\
This work & Worst case
  & $\cO\!\left(\dfrac{\alpha}{\gamma\Delta}\log\dfrac1\varepsilon\right)$
  & $\cO\!\left(\dfrac1\gamma\log\dfrac1\varepsilon\right)$ \\
\bottomrule
\end{tabular}
\end{threeparttable}
\end{table}

For an $n$-qubit system, when the block encoding $U_H$ uses $a_H$ ancilla
qubits, our expected-query algorithm uses
\[
\cO\!\left[
(a_H+1)\frac{\alpha}{\Delta}
\left(\frac1\gamma+\log\frac1\varepsilon\right)
+\frac{n+a_H+\log(2/\gamma)}{\gamma}
\right]
\]
one- and two-qubit gates in expectation, in addition to the oracle calls.
For the worst-case algorithm, the terms proportional to $1/\gamma$ gain
a factor of $\log(1/\varepsilon)$, while the high-accuracy filtering term
is unchanged.

\paragraph{Concurrent work.}
Jeffery, Somma, Witteveen, and de Wolf~\cite{JSWdW2026}
independently study ground-state preparation in the Hamiltonian-evolution
access model. They give two query-optimal algorithms in this setting,
one of which is also based on transducers.

\paragraph{Related work.}
Spectral filtering is a standard approach to ground-state preparation and
ground-state property estimation.  Such filters can be constructed using
inverse phase estimation~\cite{PW09} or linear combinations of Hamiltonian
evolutions~\cite{GTC19,KDW21}.  Polynomial filters can be implemented from a block encoding using QSVT~\cite{GSLW19,LTFilter20,LT20}, or from Hamiltonian evolution using QET-U and related quantum-signal-processing methods~\cite{DLT22,Kane24,Karacan25,DongLin24}.  Chakraborty et al.~\cite{QSVTNoBE} use Richardson extrapolation to mitigate
Trotter errors in ground-state property estimation, reducing circuit depth
and runtime.  Nearly frustration-free Hamiltonians admit improved gap dependence~\cite{TC23}.

A complementary line reduces coherent quantum resources using classical randomness.  Randomized and hybrid methods for ground-state preparation or ground-state property estimation include low-depth spectral estimators, randomized or hybrid LCU methods, stochastic QSP, and randomized QSVT~\cite{ZhangWangJohnson22,Chakraborty24,WMB24,MartynRall25,Wada25,Sun26,RandomizedQSVT}.

Ground states can also be prepared using adiabatic and eigenpath methods~\cite{Farhi01,BKS09,BKS10,AlbashLidar18}, imaginary-time evolution~\cite{McArdle19,Motta20,Gluza26,Schwartzman26}, and engineered dissipative dynamics~\cite{Verstraete09,Cubitt23,Ding24,LiZhanLin25,ZhanEtAl26,DingEtAl26}.  These approaches use different access models and assumptions from the spectral-filtering setting considered here.

When the threshold is not supplied, ground-state energy estimation can first
locate the ground energy to within a constant fraction of the
gap~\cite{LT20,LTGSEE22,DLT22}.  Randomized phase estimation provides
another approach~\cite{WBC22}.  Mande and de Wolf established upper and lower
bounds for guided eigenphase estimation~\cite{MdW26}.  Jeffery and
Witteveen~\cite{JW26} removed the $\log(1/\gamma)$ overhead using transducers.

For constant-error ground-state preparation with $\alpha=\Theta(1)$, Lin and
Tong~\cite[Theorem~10]{LT20} proved lower bounds of $\Omega(1/\gamma)$
calls to $U_H$ when $\Delta$ is constant and $\Omega(1/\Delta)$ calls when
$\gamma$ is constant.  Arihara and Murao~\cite[Corollary~7]{AM26} obtained a joint
$\Omega(1/(\gamma\Delta))$ lower bound at constant error for
$\alpha=\Theta(1)$, in a related model with query access to a block-encoding
unitary channel.  Somma and de Wolf~\cite[Section~1.2]{SdW26} further proved,
in the small-overlap regime and subject to a dimension condition, the joint
lower bound
\[
\Omega\!\left(\frac{\alpha}{\gamma\Delta}\log\frac1\varepsilon\right)
\]
for ground-state preparation to trace-distance error $\varepsilon$; their
lower bound also extends to block-encoding access.
Lin and Tong~\cite[Theorem~10]{LT20} also showed that, at constant spectral
gap and error, the $1/\gamma$ dependence of the number of calls to $U_I$
cannot be polynomially improved if the number of calls to $U_H$ is bounded
by a polynomial in $1/\gamma$.

\section{Problem setting and results}
\label{sec:model}

We now give the formal assumptions, oracle model, and output condition.  For a
register $\mathsf X$, $\ket0_{\mathsf X}$ denotes its all-zero
computational-basis state.

\subsection{Oracle and output model}

Let $\mathsf S$ be an $(N+1)$-dimensional system register, $N\ge1$, and let $H$
act on it with a unique ground state
$\ket{\psi_0}$ and spectral decomposition
\[
H=\sum_{k=0}^{N}E_k\proj{\psi_k},
\qquad
E_0<E_1\le\cdots\le E_N.
\]
We assume known parameters \(\mu\in\mathbb R\) and \(\Delta>0\) satisfying
\begin{equation}
E_0\le\mu-\frac\Delta2,
\qquad
E_1\ge\mu+\frac\Delta2.
\label{eq:separator-promise}
\end{equation}

Following~\cite[Definition~43 in the arXiv full version]{GSLW19}, a unitary
$U$ is an $\alpha$-block-encoding of an operator $A$ with $a$ ancilla qubits if
\[
(\bra{0^a}\otimes I)U(\ket{0^a}\otimes I)=\frac{A}{\alpha}.
\]
The Hamiltonian oracle $U_H$ is an $\alpha$-block-encoding of $H$ with
$a_H$ ancilla qubits.
The state-preparation oracle satisfies
\[
\ket{\phi}:=U_I\ket0,
\qquad
\beta:=|\langle\psi_0|\phi\rangle|\ge\gamma,
\qquad
0<\gamma\le1.
\]
We call any Hamiltonian and pair of oracles satisfying these conditions a
valid instance.  A call to $U_H$, $U_H^\dagger$, or a controlled version
counts as one Hamiltonian query; a call to $U_I$, $U_I^\dagger$, or a
controlled version counts as one state-preparation query.  We write $q_H,q_I$
for the numbers of these calls in a run, which may be random, and $Q_H,Q_I$
for worst-case bounds that hold on every run.

A circuit run outputs a system state and may also return a classical flag
$b\in\{0,1\}$, where $b=1$ denotes success.  We write $\rho$ for the system
state averaged over measurement outcomes and classical randomness, and
$\rho_{\rm succ}$ for the state conditioned on $b=1$.  Error is measured by
\[
\Dtr(\rho,\sigma):=\frac12\norm{\rho-\sigma}_1.
\]
An algorithm with random query counts is required to output a state with
probability one.
In asymptotic bounds, $\alpha,\gamma,\Delta$ may vary with $N$.

\subsection{Main results}

We first give an algorithm with the following expected query bounds
and prove a matching lower bound on $\mathbb E[q_H]$.

\begin{theorem}[Expected query complexity]
\label{thm:expected-query-complexity}
Let $0<\gamma\le1/\sqrt2$, $0<\varepsilon\le1/10$, and
$0<\Delta\le\alpha/8$.  There is an algorithm that outputs $\rho$ satisfying
$\Dtr(\rho,\proj{\psi_0})\le\varepsilon$.  Its expected query counts are
\[
\mathbb E[q_H]=\cO\!\left(
\frac{\alpha}{\gamma\Delta}
+\frac{\alpha}{\Delta}\log\frac1\varepsilon
\right),
\qquad
\mathbb E[q_I]=\cO\!\left(\frac1\gamma\right).
\]
Conversely, any algorithm meeting this guarantee on all valid instances
satisfies
\[
\mathbb E[q_H]=\Omega\!\left(
\frac{\alpha}{\gamma\Delta}
+\frac{\alpha}{\Delta}\log\frac1\varepsilon
\right)
\]
on some instance, even with unrestricted access to $U_I$.  Moreover, if
$\mathbb E[q_H]=o((\alpha/\Delta)\sqrt N)$ uniformly over valid
$(N+1)$-dimensional instances, then for all sufficiently large $N$ some
valid instance satisfies
\[
\mathbb E[q_I]=\Omega(1/\gamma).
\]
\end{theorem}

The upper bounds are proved in
Section~\ref{subsec:constant-success-preparation}.
The lower bounds on $\mathbb E[q_H]$ and $\mathbb E[q_I]$ are proved in
Sections~\ref{subsec:hamiltonian-query-lower}
and~\ref{subsec:initial-query-lower}, respectively.

For query counts bounded on every run, we give an algorithm with the
same output guarantee that matches the lower bound on $Q_H$ of
Somma and de Wolf~\cite{SdW26}.

\begin{theorem}[Worst-case query complexity]
\label{thm:unconditional-output}
Let $0<\gamma\le1/\sqrt2$, $0<\varepsilon\le1/10$, and
$0<\Delta\le\alpha/8$.  There is an algorithm that, on every valid instance,
outputs $\rho$ satisfying $\Dtr(\rho,\proj{\psi_0})\le\varepsilon$.
Its worst-case query counts satisfy
\[
Q_H=\cO\!\left(
\frac{\alpha}{\gamma\Delta}\log\frac1\varepsilon
\right),
\qquad
Q_I=\cO\!\left(\frac1\gamma\log\frac1\varepsilon\right).
\]
Conversely, any algorithm meeting this guarantee on all valid instances
satisfies
\[
Q_H=\Omega\!\left(
\frac{\alpha}{\gamma\Delta}\log\frac1\varepsilon
\right)
\]
on some instance.  This bound follows from the lower bound of Somma and de
Wolf~\cite[Section~1.2]{SdW26} and Lemma~\ref{lem:precision-lower}.
Moreover, if $Q_H=o((\alpha/\Delta)\sqrt N)$ on valid
$(N+1)$-dimensional instances, then for all sufficiently large $N$ some
valid instance satisfies
\[
Q_I=\Omega\!\left(\frac1\gamma\log\frac1\varepsilon\right).
\]
\end{theorem}

The upper bounds are proved in Section~\ref{subsec:bounded-query-preparation}.
The lower bounds on $Q_H$ and $Q_I$ are proved in
Sections~\ref{subsec:hamiltonian-query-lower}
and~\ref{subsec:initial-query-lower}, respectively.

More generally, we allow the algorithm to report failure with probability
at most $\zeta$.  Conditioned on success, its output state is within trace
distance $\varepsilon$ of $\proj{\psi_0}$.

\begin{theorem}[Upper bound with a failure flag]
\label{thm:upper-bound}
Let $0<\gamma\le1/\sqrt2$, $0<\varepsilon\le\zeta\le1/10$, and
$0<\Delta\le\alpha/8$.  There is an algorithm that, on every valid instance,
succeeds with probability at least $1-\zeta$ and,
conditioned on success, satisfies
$\Dtr(\rho_{\rm succ},\proj{\psi_0})\le\varepsilon$.  Its worst-case query
counts satisfy
\[
Q_H=\cO\!\left(
\frac{\alpha}{\gamma\Delta}\log\frac1\zeta
+\frac{\alpha}{\Delta}\log\frac1\varepsilon
\right),
\qquad
Q_I=\cO\!\left(\frac1\gamma\log\frac1\zeta\right).
\]
\end{theorem}

The upper bound is proved in Section~\ref{subsec:bounded-query-preparation}.
Taking the error target and $\zeta$ both to be $\varepsilon/2$ and returning
an arbitrary state on failure gives the worst-case upper bounds in
Theorem~\ref{thm:unconditional-output}.

A lower bound on $Q_I$ must account for algorithms that do not call $U_I$.
For an $(N+1)$-dimensional system maximally entangled with a reference, the
reduced system state is $I/(N+1)$, so
$\bra{\psi_0}(I/(N+1))\ket{\psi_0}=1/(N+1)$.  Spectral filtering can
therefore prepare $\ket{\psi_0}$ without $U_I$, at a cost in calls to $U_H$
that grows with $N$.  The following theorem quantifies the tradeoff.

\begin{theorem}[Query tradeoff]
\label{thm:initial-query-tradeoff}
Fix \(0<s_0\le1\), \(0<\varepsilon<1\),
\(0<\gamma\le1/\sqrt2\), $\alpha>0$, and $0<\Delta<2\alpha$.
For every integer $N\ge4/(s_0(1-\varepsilon))$, there is a family of
\((N+1)\)-dimensional valid instances such that every
algorithm that succeeds with probability at least $s_0$ and satisfies
$\Dtr(\rho_{\rm succ},\proj{\psi_0})\le\varepsilon$ on every member of the
family obeys
\begin{equation}
\gamma Q_I+\frac{\Delta Q_H}{2\alpha\sqrt N}
\ge\frac{s_0(1-\varepsilon)}{4\sqrt2}.
\label{eq:initial-H-tradeoff}
\end{equation}
\end{theorem}

When $Q_H=o((\alpha/\Delta)\sqrt N)$, the Hamiltonian-query term
in~\eqref{eq:initial-H-tradeoff} vanishes as $N$ grows.
The tradeoff then implies the following lower bound on $Q_I$.
\begin{corollary}[Lower bound on $Q_I$]
\label{cor:initial-query-lower}
Fix $0<s_0\le1$ and $0<\varepsilon\le1/10$.  For each $N$, let
$\alpha,\gamma,\Delta$ be in the ranges of
Theorem~\ref{thm:initial-query-tradeoff}.  Suppose an algorithm succeeds
with probability at least $s_0$ and satisfies
$\Dtr(\rho_{\rm succ},\proj{\psi_0})\le\varepsilon$ on every valid
$(N+1)$-dimensional instance.  If
$Q_H=o((\alpha/\Delta)\sqrt N)$, then for all sufficiently large $N$ some
valid instance satisfies
$Q_I=\Omega_{s_0}(1/\gamma)$.
\end{corollary}

Theorem~\ref{thm:initial-query-tradeoff} and
Corollary~\ref{cor:initial-query-lower} are proved in
Section~\ref{subsec:initial-query-lower}.

\section{Algorithm overview}
\label{sec:overview}

Our algorithms use amplitude amplification to obtain a state with constant
overlap with $\ket{\psi_0}$, then apply spectral filtering to reach error
$\varepsilon$.  The key is to use transducer composition in the amplification
stage, avoiding the $\log(1/\gamma)$ overhead from approximating each
ground-state reflection separately by QSVT.  We first describe spectral
filtering and amplitude amplification in terms of reflections, then explain
the transducer construction and the final filtering steps.

After shifting the threshold to zero and rescaling as in
Lemma~\ref{lem:hamiltonian-rescaling}, we write $H$ and $E_k$ for the
normalized Hamiltonian and its eigenvalues throughout this section.  They
satisfy
\begin{equation}
\norm{H}\le1,
\qquad
E_0\le-\delta,
\qquad
E_k\ge\delta\quad\text{for }k\ge1,
\qquad
\delta^{-1}=\cO(\alpha/\Delta).
\label{eq:normalized-promise}
\end{equation}

\subsection{Spectral filtering and approximate reflection}

The promise~\eqref{eq:normalized-promise} places the ground eigenvalue in
$[-1,-\delta]$ and all excited eigenvalues in $[\delta,1]$.  For
$0<\eta<1/2$, QSVT~\cite{GSLW19,LT20} provides a spectral filter
$U_{\eta,\delta}$ acting on the system and an ancilla register $\mathsf F$.
This unitary satisfies
\begin{equation}
U_{\eta,\delta}\ket0_{\mathsf F}\ket{\psi_k}
=f(E_k)\ket0_{\mathsf F}\ket{\psi_k}
+\ket{\omega_{k,\perp}},
\qquad
(\bra0_{\mathsf F}\otimes I)\ket{\omega_{k,\perp}}=0.
\label{eq:overview-filter}
\end{equation}
The polynomial $f$ satisfies
\begin{equation}
|f(E_0)-1|\le\eta,
\qquad
|f(E_k)|\le\eta,\qquad k\ge1.
\label{eq:overview-filter-bounds}
\end{equation}
Implementing $U_{\eta,\delta}$ uses
$\cO((\alpha/\Delta)\log(1/\eta))$ calls to $U_H$ and $U_H^\dagger$.
On outcome $\ket0_{\mathsf F}$, the spectral filter approximately preserves
the ground-state component of the input and suppresses the excited-state
components.

The spectral filter also gives an approximate ground-state reflection:
apply the filter, reflect the $\mathsf F=0$ subspace, and apply the inverse
filter.  Define $R_g:=2\proj{\psi_0}-I$.  Using $U_{\eta^2,\delta}$, set
\[
\widetilde R_g:=U_{\eta^2,\delta}^\dagger
\bigl(2(\proj0_{\mathsf F}\otimes I)-I\bigr)U_{\eta^2,\delta}.
\]
On inputs with $\mathsf F$ initialized to $\ket0_{\mathsf F}$, the bounds
in Eq.~\eqref{eq:overview-filter-bounds} with $\eta$ replaced by $\eta^2$,
together with the unitarity of $U_{\eta^2,\delta}$, imply
\begin{equation}
\norm{\bigl(\widetilde R_g-I_{\mathsf F}\otimes R_g\bigr)
(\proj0_{\mathsf F}\otimes I)}=\cO(\eta).
\label{eq:overview-approximate-reflection}
\end{equation}
Using $\eta^2$ instead of $\eta$ changes only the constant in the logarithm,
so $\widetilde R_g$ uses
$\cO((\alpha/\Delta)\log(1/\eta))$ calls to $U_H$ and $U_H^\dagger$.

\subsection{Amplitude amplification and error accumulation}
\label{subsec:overview-amplification}

To amplify the ground-state component of $\ket\phi$, we alternate $R_g$
with the reflection about $\ket\phi$, as in amplitude
amplification~\cite{BHMT02}.  We first describe the resulting rotation, then
bound the error incurred by replacing $R_g$ with $\widetilde R_g$.
Choose the global phase of $\ket{\psi_0}$ so
that $\langle\psi_0|\phi\rangle=\beta$, and write
\[
\ket\phi=\sin\theta\ket{\psi_0}+\cos\theta\ket{\phi_\perp},
\qquad
\sin\theta=\beta,
\qquad
\langle\psi_0|\phi_\perp\rangle=0,
\]
where $\ket{\phi_\perp}$ is normalized.  Define
\[
R_\phi:=2\proj\phi-I=U_I(2\proj0-I)U_I^\dagger,
\qquad
G:=R_\phi R_g.
\]
Both $R_\phi$ and $R_g$ preserve
$\operatorname{span}\{\ket{\psi_0},\ket{\phi_\perp}\}$.  Direct calculation
on this subspace gives
\begin{equation}
G^j\ket\phi
=(-1)^j\left[
\sin((2j+1)\theta)\ket{\psi_0}
+\cos((2j+1)\theta)\ket{\phi_\perp}
\right].
\label{eq:overview-grover-rotation}
\end{equation}
Thus a suitable $j=\cO(1/\beta)$ gives
$|\bra{\psi_0}G^j\ket\phi|^2=\Omega(1)$.
Since only $\beta\ge\gamma$ is known, an iteration count chosen for
$\beta=\gamma$ may overshoot: for larger $\beta$, the angle
$(2j+1)\theta$ can pass $\pi/2$, where the ground-state overlap
reaches its first maximum.  Randomizing the iteration count over a
range of size $\cO(1/\gamma)$ avoids overshooting~\cite{BBHT98}.

As a baseline, we can implement the ideal rotation in
Eq.~\eqref{eq:overview-grover-rotation} to error $\cO(\varepsilon)$ by
replacing each occurrence
of $R_g$ in $G^j=(R_\phi R_g)^j$ with its QSVT approximation
$\widetilde R_g$, using a separate ancilla register $\mathsf F_t$ for the
$t$th occurrence.  Let $\ket{\Phi_j}$ denote
the output of this circuit on $\ket\phi$, including these registers.  Each of
the $j$ replacements changes the output by $\cO(\eta)$ by
Eq.~\eqref{eq:overview-approximate-reflection}, and the remaining operations
are unitary.  A telescoping sum therefore gives
\begin{equation}
\norm{\ket{\Phi_j}
-\ket0_{\mathsf F_1\cdots\mathsf F_j}G^j\ket\phi}
=\cO(j\eta).
\label{eq:overview-hybrid}
\end{equation}
Thus $j=\cO(1/\gamma)$ iterations turn a per-reflection error
$\cO(\eta)$ into a total error $\cO(\eta/\gamma)$.
Taking $\eta=\cO(\gamma\varepsilon)$ keeps
the accumulated error within $\cO(\varepsilon)$.
Together with the cost of $\widetilde R_g$ and
$\cO(1/\gamma)$ amplification steps, this recovers the query scaling of
Lin and Tong~\cite[Theorem~6]{LT20}:
\[
Q_H=\cO\!\left(
\frac{\alpha}{\gamma\Delta}
\log\frac1{\gamma\varepsilon}
\right).
\]

Our algorithm instead uses two stages.  In the first, amplitude amplification
only needs to produce a state $\rho_c$ satisfying
$\bra{\psi_0}\rho_c\ket{\psi_0}=\Omega(1)$.  The total error in
\eqref{eq:overview-hybrid} therefore need only be a sufficiently small
constant, which allows $\eta=\Theta(\gamma)$ for $j=\cO(1/\gamma)$.
In the second stage, one application of $U_{\eta,\delta}$ to $\rho_c$ with
$\eta=\Theta(\varepsilon)$ yields the outcome $\ket0_{\mathsf F}$ with
constant probability.  Conditioned on this outcome, the system state is within
trace distance $\cO(\varepsilon)$ of $\proj{\psi_0}$.  The two stages use
\[
\cO\!\left(
\frac{\alpha}{\gamma\Delta}\log\frac1\gamma
+\frac{\alpha}{\Delta}\log\frac1\varepsilon
\right)
\]
calls to $U_H$ and $U_H^\dagger$ in total.  The factor $\log(1/\gamma)$ in the
first term comes from approximating each of the $\cO(1/\gamma)$ ground-state
reflections by QSVT to error $\cO(\gamma)$.

\subsection{Ground-state reflection by a transducer}
\label{subsec:overview-transducer}

We remove the remaining $\log(1/\gamma)$ factor by using transducers to
compose the reflections.  Following~\cite{BJY24}, we encode a unitary $U$
on a public space $\mathcal H$ using a transducer $S$, a unitary acting on
$\mathcal H\oplus\mathcal L$, where $\mathcal L$ is the private space.
For each public input $\xi\in\mathcal H$, a catalyst $v\in\mathcal L$ satisfies
\[
S(\xi\oplus v)=U\xi\oplus v.
\]
We write $\xi\transduce{S}U\xi$ for this relation.  The unitary $U$ is the
transduction action of $S$, and $W(S,\xi):=\norm{v}^2$ is the transduction
complexity~\cite[Section~7.1]{BJY24}.  The catalyst is not
supplied with the public input, so one execution of $S$ need not map $\xi$ to
$U\xi$.  For normalized $\xi$ with $W(S,\xi)\le W_0$, the circuit
implementation theorem of~\cite[Theorem~3.2]{BJY24} gives a circuit that
approximates $U\xi$ to Euclidean error $\varepsilon_{\rm c}$ using
\[
\cO\!\left(1+\frac{W_0}{\varepsilon_{\rm c}^2}\right)
\]
controlled executions of $S$.

In the first stage, we construct transducers $S_g$ and $S_\phi$ for the
reflections $R_g$ and $R_\phi$, respectively.  For every normalized system
state $\ket\psi$, $S_g$ satisfies
\[
\ket0\ket\psi \transduce{S_g} \ket0 R_g\ket\psi,
\qquad
W(S_g,\ket0\ket\psi)=\cO(1).
\]
Here $\ket0$ denotes the all-zero state of the ancillas in the public space
of $S_g$.  The transduction action on these inputs is exactly $R_g$ on the
system.
Proposition~\ref{prop:ground-reflection} constructs $S_g$ from the
constant-accuracy filter $U_{1/10,\delta}$, the purifier of Belovs and
Jeffery~\cite{BJ26}, and the composition results of~\cite{BJY24}.
One execution of $S_g$ makes $\cO(1)$ controlled calls to
$U_{1/10,\delta}$ and $U_{1/10,\delta}^\dagger$.
We obtain $S_\phi$ by converting the two-query circuit for $R_\phi$;
its transduction complexity is constant
(Lemma~\ref{lem:circuit-to-transducer}).

We use sequential composition~\cite[Proposition~9.10]{BJY24} to combine $j$
copies each of $S_g$ and $S_\phi$ into a transducer whose transduction action
maps $\ket0\ket\phi$ to $\ket0 G^j\ket\phi$.  Its transduction complexity
is the sum of the component complexities plus $\cO(j)$, hence $\cO(j)$.
Since the transduction actions of $S_g$ and $S_\phi$ are exactly $R_g$ and
$R_\phi$, respectively, their composition introduces no approximation
error.  We apply the circuit implementation
theorem once to the complete transducer.  On input $\ket0\ket\phi$, the
resulting circuit approximates
$\ket0 G^j\ket\phi$ to sufficiently small constant error using $\cO(j)$ calls to
$U_{1/10,\delta}$ and
$U_{1/10,\delta}^\dagger$.  Each such call uses
$\cO(\alpha/\Delta)$ calls to $U_H$ and $U_H^\dagger$.  Since each sampled
iteration count satisfies $j=\cO(1/\gamma)$, the circuit uses
\[
\cO\!\left(\frac{\alpha}{\gamma\Delta}\right)
\]
calls to $U_H$ and $U_H^\dagger$.

The first stage either leaves $\ket\phi$ unchanged or samples $j$ as in
Section~\ref{subsec:overview-amplification} and runs the circuit for
$G^j\ket\phi$.
With ideal reflections, the average squared overlap with
$\ket{\psi_0}$ is at least $3/16$
(Lemma~\ref{lem:ideal-averaged-grover-state}).  The constant-error circuit
implementation therefore yields an average system state $\rho_c$ satisfying
\[
p:=\bra{\psi_0}\rho_c\ket{\psi_0}=\Omega(1).
\]
Preparing $\ket\phi$ and implementing the $R_\phi$ reflections use
$\cO(1/\gamma)$ calls to $U_I$ and $U_I^\dagger$ in this stage.

\subsection{Final filtering}

Having prepared $\rho_c$ with constant overlap with the ground state, we now
use spectral filtering to reach error $\varepsilon$.  We apply the spectral
filter $U_{\eta,\delta}$ from Eq.~\eqref{eq:overview-filter} to $\rho_c$ with
$\eta=\Theta(\varepsilon)$, and declare
success on the outcome $\ket0_{\mathsf F}$.  For this outcome, the ground-state
contribution to the probability is at least $(1-\eta)^2p$, while the total
excited-state contribution is at most $\eta^2(1-p)$.
Since $p=\Omega(1)$, success occurs with constant probability,
and the normalized system state has trace distance
$\cO(\eta/\sqrt p)=\cO(\varepsilon)$ from $\proj{\psi_0}$
(Proposition~\ref{prop:final-filter-correctness}).  The final filter uses
$\cO((\alpha/\Delta)\log(1/\varepsilon))$ calls to $U_H$ and $U_H^\dagger$ and
no calls to $U_I$ or $U_I^\dagger$
(Proposition~\ref{prop:final-filter-cost}).  For one run with constant
success probability, the worst-case query counts therefore satisfy
\[
Q_I=\cO\!\left(\frac1\gamma\right),
\qquad
Q_H=\cO\!\left(
\frac{\alpha}{\gamma\Delta}
+\frac{\alpha}{\Delta}\log\frac1\varepsilon
\right).
\]
Repeating independent runs until success gives the expected-query upper
bounds in Theorem~\ref{thm:expected-query-complexity}.

To reduce the failure probability to $\zeta$, repeating the two-stage algorithm
would also multiply the high-accuracy filtering cost by
$\log(1/\zeta)$.  We avoid this multiplication by applying spectral filters
of increasing accuracy to each newly prepared $\rho_c$, ending with a filter
of error $\cO(\varepsilon)$.  An attempt restarts at the first rejection and
succeeds only if all filters accept.  We square the error target from one
filter to the next, so their query costs grow geometrically.  The earlier
filters suppress the excited-state components before the costly filters are
reached, while the later filters nearly preserve the ground-state component.
A first rejection at a later filter therefore becomes rapidly less likely.
This bounds, with high probability, the total number of calls to $U_H$ and
$U_H^\dagger$ made by the spectral filters in rejected attempts.  The filters
in a successful attempt use $\cO((\alpha/\Delta)\log(1/\varepsilon))$ such
calls, after which the algorithm stops.  Limiting both the number of attempts
and the total number of calls made by the filters gives failure probability
at most $\zeta$ and worst-case query counts satisfying
\[
Q_H=\cO\!\left(
\frac{\alpha}{\gamma\Delta}\log\frac1\zeta
+\frac{\alpha}{\Delta}\log\frac1\varepsilon
\right),
\qquad
Q_I=\cO\!\left(\frac1\gamma\log\frac1\zeta\right).
\]
These are the bounds in Theorem~\ref{thm:upper-bound} and the full analysis is
given in Proposition~\ref{prop:increasing-accuracy-filtering}.

\section{Preliminaries}

\subsection{Shifting and rescaling the Hamiltonian}

We shift the known threshold to zero and rescale the spectrum to obtain the
promise used in~\eqref{eq:normalized-promise}.

\begin{lemma}[Shifting and rescaling]
\label{lem:hamiltonian-rescaling}
Under the access model of Section~\ref{sec:model}, set
\[
\widetilde H:=\frac{H-\mu I}{\bar\alpha},
\qquad
\bar\alpha:=\alpha+|\mu|.
\]
It has the same eigenvectors as $H$, with eigenvalues
$\widetilde E_k=(E_k-\mu)/\bar\alpha$.  Set
\[
\delta:=\min\!\left\{
\frac{\Delta}{2\bar\alpha},\frac12
\right\}.
\]
Then
\[
\norm{\widetilde H}\le1,
\qquad
\widetilde E_0\le-\delta,
\qquad
\widetilde E_k\ge\delta\text{ for }k\ge1,
\qquad
\delta^{-1}\le\frac{4\alpha}{\Delta}.
\]
Moreover, a $1$-block-encoding of $\widetilde H$ can be implemented with one
controlled call to $U_H$, and its adjoint with one controlled call to
$U_H^\dagger$.
\end{lemma}

\begin{proof}
The threshold promise~\eqref{eq:separator-promise} and $\norm{H}\le\alpha$ imply
\[
-\alpha\le E_0<\mu<E_1\le\alpha.
\]
Consequently, $|\mu|<\alpha$ and $\bar\alpha<2\alpha$.  Moreover, the
triangle inequality gives
\[
\norm{\widetilde H}
\le\frac{\norm{H}+|\mu|}{\bar\alpha}
\le1.
\]
The standard linear-combination construction for block-encodings
\cite[Lemma~52 in the arXiv full version]{GSLW19}, applied to $U_H$ and the identity, gives a
$\bar\alpha$-block-encoding of $H-\mu I$, or equivalently a
$1$-block-encoding of $\widetilde H$.
The identity term requires no oracle call, so the claimed query costs follow.

Since $\widetilde H$ is an affine function of $H$, it has the same
eigenvectors as $H$, with eigenvalues $\widetilde E_k=(E_k-\mu)/\bar\alpha$.  The
same promise therefore yields
\[
\widetilde E_0\le-\frac{\Delta}{2\bar\alpha},
\qquad
\widetilde E_k\ge\frac{\Delta}{2\bar\alpha}
\text{ for } k\ge1.
\]
Because $\delta\le\Delta/(2\bar\alpha)$, these inequalities give
$\widetilde E_0\le-\delta$ and $\widetilde E_k\ge\delta$ for every $k\ge1$.
The cap at $1/2$ ensures $0<\delta<1$, as required by
Lemma~\ref{lem:spectral-filter}.

Finally, the threshold promise~\eqref{eq:separator-promise} and
$\norm{H}\le\alpha$ give
\[
\Delta\le E_1-E_0\le2\alpha.
\]
Since $\bar\alpha<2\alpha$, both $2\bar\alpha/\Delta$ and $2$ are at most
$4\alpha/\Delta$.  Hence
\[
\delta^{-1}
=\max\!\left\{\frac{2\bar\alpha}{\Delta},2\right\}
\le\frac{4\alpha}{\Delta}.
\]
This proves the lemma.
\end{proof}

Until Section~\ref{sec:optimality}, $H$ and $E_k$ denote the normalized
operator $\widetilde H$ and its eigenvalues $\widetilde E_k$.  Queries to its
block-encoding are charged to the original oracle according to
Lemma~\ref{lem:hamiltonian-rescaling}.  The linear-combination construction
adds at most one ancilla qubit; this constant increase is absorbed in the
$a_H+1$ factors in the gate bounds below.

\subsection{Spectral filtering by QSVT}

We use the following standard spectral-filtering consequence of quantum
singular value transformation.

\begin{lemma}[QSVT spectral filter~{\cite[Lemma~5]{LT20}}]
\label{lem:spectral-filter}
Let $H$ be Hermitian with $\norm H\le1$, and suppose that a $1$-block-encoding
of $H$ with $a_H$ ancilla qubits is available.  For $0<\eta<1/2$ and
$0<\delta<1$, there is a real polynomial $f$ with $|f(x)|\le1$ for
$x\in[-1,1]$ such that
\[
\begin{cases}
|f(x)-1|\le\eta,&x\in[-1,-\delta],\\
|f(x)|\le\eta,&x\in[\delta,1].
\end{cases}
\]
QSVT gives a unitary $U_{\eta,\delta}$ acting on the system and an ancilla
register $\mathsf F$ of $a_H+\cO(1)$ qubits such that
\[
(\bra0_{\mathsf F}\otimes I)U_{\eta,\delta}
(\ket0_{\mathsf F}\otimes I)=f(H).
\]
The circuit uses
$\cO(\delta^{-1}\log(1/\eta))$ queries to the block-encoding and
$\cO((a_H+1)\delta^{-1}\log(1/\eta))$ one- and two-qubit gates.
\end{lemma}

\subsection{Transducers}

We recall the transducer formalism of~\cite{BJY24}.  All direct sums below
are orthogonal direct sums.

\begin{definition}[Transducer]
\label{def:transducer}
Let $\mathcal H$ and $\mathcal L$ be Hilbert spaces.  A \emph{transducer with
public space $\mathcal H$ and private space $\mathcal L$} is a unitary
\[
S\colon\mathcal H\oplus\mathcal L\longrightarrow
\mathcal H\oplus\mathcal L.
\]
Given $\xi,\tau\in\mathcal H$ and $v\in\mathcal L$ satisfying
\begin{equation}
S(\xi\oplus v)=\tau\oplus v,
\label{eq:transduction-general}
\end{equation}
we say that $S$ \emph{transduces} $\xi$ into $\tau$, write
$\xi\transduce{S}\tau$, and call $v$ a \emph{catalyst}.
A unitary $U\colon\mathcal H\to\mathcal H$ is the \emph{transduction action}
of $S$ if $\xi\transduce{S}U\xi$ for every $\xi\in\mathcal H$.  We write
$S\DownTransduce_{\mathcal H}=U$~\cite[Definition~2.5]{BJ26}.
\end{definition}

\begin{definition}[Transduction complexity~{\cite[Section~7.1]{BJY24}}]
\label{def:transduction-complexity}
For a transduction $\xi\transduce{S(O)}\tau$ with catalyst $v$, its
transduction complexity is
\[
W(S,O,\xi):=\norm{v}^2.
\]
We use the catalysts specified by the constructions below.
\end{definition}

We abbreviate $S(O)$ to $S$ and omit fixed arguments of $W$ when the context
is clear.

The following result explains how transduction complexity controls the cost
of implementing the transduction action. The catalyst $v$ for $\xi\transduce{S}\tau$ is not an input to this circuit.

\begin{lemma}[Circuit implementation~{\cite[Theorem~3.2]{BJY24}}]
\label{lem:circuit-implementation}
Let $S$ be a transducer, $W_0\ge0$, and $\varepsilon_{\rm c}>0$.
There is a circuit that, for every normalized $\xi\in\mathcal H$ with
$W(S,\xi)\le W_0$, starts from $\xi\oplus0$ and produces a state within
Euclidean distance $\varepsilon_{\rm c}$ of $\tau\oplus0$, where
$\xi\transduce{S}\tau$.  The circuit uses
\[
K=\cO(1+W_0/\varepsilon_{\rm c}^2)
\]
controlled executions of $S$ and $\cO(K)$ additional one- and two-qubit
gates.
\end{lemma}

The composition results below use canonical transducers, whose input oracles
act only on part of the private space.

\begin{definition}[Canonical transducer]
\label{def:canonical-transducer}
Let $O$ act on $\mathcal M$.  A transducer on
$\mathcal H\oplus\mathcal L$ is canonical if a subspace of $\mathcal L$
has the form $\mathcal E\otimes\mathcal M$ for some Hilbert space
$\mathcal E$, and
\begin{equation}
S(O)=S^\circ\widetilde O,
\qquad
\widetilde O|_{\mathcal E\otimes\mathcal M}=I_{\mathcal E}\otimes O,
\qquad
\widetilde O|_{(\mathcal E\otimes\mathcal M)^\perp}=I,
\label{eq:canonical-form}
\end{equation}
where the orthogonal complement is taken in $\mathcal H\oplus\mathcal L$
and $S^\circ$ is an oracle-independent unitary.  We call $S^\circ$ the
\emph{work unitary}.
\end{definition}

For an oracle-independent unitary $V$, let $T(V)$ denote the number of
one- and two-qubit gates used to implement $V$; for a canonical transducer
$S$, write $T(S):=T(S^\circ)$.

The following lemma converts a unitary circuit with oracle access into a
canonical transducer whose transduction action is the unitary implemented
by the circuit.  The result is a circuit-model consequence of
\cite[Theorem~3.5 and Figure~3.6]{BJY24}.
\begin{lemma}[Circuit-to-transducer conversion]
\label{lem:circuit-to-transducer}
Let $A(O)\colon\mathcal H\to\mathcal H$ be a unitary circuit with $N$
queries to $O$ and $T_A$ one- and two-qubit gates outside the oracle calls.
There is a canonical transducer $S_A(O)$ with public space $\mathcal H$ such
that, for every $\xi\in\mathcal H$,
\begin{equation}
\xi\transduce{S_A(O)}A(O)\xi,
\qquad
W(S_A,O,\xi)\le N\norm{\xi}^2.
\label{eq:circuit-to-transducer}
\end{equation}
Moreover, $T(S_A)=\cO(T_A+N)$.
\end{lemma}

For the construction in~\cite{BJY24}, $W(S_A,O,\xi)$ is the sum of the
squared norms of the components sent to the oracle.  There are $N$ such
components, each of norm at most $\norm{\xi}$.  The circuit has $T_A+N$
gates in total.

Canonical transducers admit sequential and functional
composition~\cite[Section~9]{BJY24}.  Under sequential composition, the
transduction actions are applied in order; the component transduction
complexities and the squared norms of the intermediate public states add.
Functional composition is used when the transduction
action of one transducer is the input oracle of another.  We record the two
forms used below and include their proofs for completeness.

\begin{proposition}[Sequential composition~{\cite[Definition~9.8 and
Proposition~9.10]{BJY24}}]
\label{prop:sequential-composition}
Let $S_1,\ldots,S_m$ be canonical transducers with the same public space
$\mathcal H$ and input oracle $O$.  Suppose that
\[
\xi_t\transduce{S_t(O)}\xi_{t+1},
\qquad t=1,\ldots,m.
\]
Writing $U_t:=S_t(O)\DownTransduce_{\mathcal H}$, their sequential composition
$S_m*\cdots*S_1$ has transduction action $U_m\cdots U_1$ and transduces
$\xi_1$ into $\xi_{m+1}$.  Its parallel implementation satisfies
\begin{equation}
W(S_m*\cdots*S_1,O,\xi_1)
=\sum_{t=2}^m\norm{\xi_t}^2
+\sum_{t=1}^m W(S_t,O,\xi_t).
\label{eq:sequential-composition}
\end{equation}
The same implementation has gate count
\begin{equation}
T(S_m*\cdots*S_1)
=\cO\!\left(\log m+T\!\left(\bigoplus_{t=1}^m S_t^\circ\right)\right).
\label{eq:sequential-composition-gate-cost}
\end{equation}
\end{proposition}

\begin{proof}
Let $O$ act on $\mathcal M$.  Write
$S_t(O)=S_t^\circ\widetilde O_t$, and let $\mathcal L_t$ be its
private space.  Take copies $\mathcal H_1,\ldots,\mathcal H_m$ of
$\mathcal H$, with $\mathcal H_{m+1}=\mathcal H_1$.  The transducer
$S_m*\cdots*S_1$ acts on
\[
\bigoplus_{t=1}^m(\mathcal H_t\oplus\mathcal L_t),
\]
where $\mathcal H_1$ is public and all other summands are private.  Let
$P$ send the copy of $x\in\mathcal H$ in $\mathcal H_t$ to its copy in
$\mathcal H_{t+1}$, and let $P$ act as the identity on each
$\mathcal L_t$.  The parallel construction is
\begin{equation}
S_m*\cdots*S_1
=P\left(\bigoplus_{t=1}^m S_t^\circ\right)
\left(\bigoplus_{t=1}^m\widetilde O_t\right).
\label{eq:sequential-composition-construction}
\end{equation}
For each $t$, the canonical form of $S_t$ gives a subspace
$\mathcal E_t\otimes\mathcal M\subseteq\mathcal L_t$ on which
$\widetilde O_t=I_{\mathcal E_t}\otimes O$; it is the identity on the rest
of $\mathcal H_t\oplus\mathcal L_t$.  Set
$\mathcal E:=\bigoplus_{t=1}^m\mathcal E_t$.  Since
\[
\bigoplus_{t=1}^m(\mathcal E_t\otimes\mathcal M)
\cong\mathcal E\otimes\mathcal M,
\]
the operator $\bigoplus_{t=1}^m\widetilde O_t$ acts as
$I_{\mathcal E}\otimes O$ on this private subspace and as the identity
on its orthogonal complement.  Hence $S_m*\cdots*S_1$ has the canonical
form~\eqref{eq:canonical-form}, with one controlled call to $O$.

For each $t$, use the catalyst $v_t\in\mathcal L_t$ of
$\xi_t\transduce{S_t(O)}\xi_{t+1}$, so
$S_t(O)(\xi_t\oplus v_t)=\xi_{t+1}\oplus v_t$ and
$\norm{v_t}^2=W(S_t,O,\xi_t)$.  Put $\xi_t$ in
$\mathcal H_t$ for $t\ge2$ and define a catalyst for
$S_m*\cdots*S_1$ by
\[
v=\bigoplus_{t=2}^m\xi_t\ \oplus\
  \bigoplus_{t=1}^m v_t.
\]
The component of $\xi_1\oplus v$ in $\mathcal H_t\oplus\mathcal L_t$
is $\xi_t\oplus v_t$.  Hence
\[
\begin{aligned}
(S_m*\cdots*S_1)(O)(\xi_1\oplus v)
&=P\left(\bigoplus_{t=1}^m S_t(O)(\xi_t\oplus v_t)\right)\\
&=P\left(\bigoplus_{t=1}^m(\xi_{t+1}\oplus v_t)\right)\\
&=\xi_{m+1}\oplus\left(\bigoplus_{t=2}^m\xi_t\right)
  \oplus\left(\bigoplus_{t=1}^m v_t\right)\\
&=\xi_{m+1}\oplus v.
\end{aligned}
\]
Since $\norm{v}^2=\sum_{t=2}^m\norm{\xi_t}^2+
\sum_{t=1}^m\norm{v_t}^2$, this proves
\eqref{eq:sequential-composition}; the transduction action is $U_m\cdots U_1$.

Encode $t$ in $\cO(\log m)$ qubits and use a bit to distinguish
$\mathcal H_t$ from $\mathcal L_t$.  Then $P$ increments $t$ modulo
$m$ when the component is in $\mathcal H_t$.  Ripple-carry arithmetic
implements this controlled increment modulo $m$ with $\cO(\log m)$
one- and two-qubit gates~\cite[Sections~4.1 and~4.3]{CDKM04}.
The work unitary in
\eqref{eq:sequential-composition-construction} is
$P(\bigoplus_t S_t^\circ)$, giving
\eqref{eq:sequential-composition-gate-cost}.
\end{proof}

For a canonical transducer $S(O)$ on $\mathcal H\oplus\mathcal L$,
$I_{\mathcal E}\otimes S(O)$ is canonical with public space
$\mathcal E\otimes\mathcal H$ and private space
$\mathcal E\otimes\mathcal L$~\cite[Definition~9.6 and Corollary~9.7]{BJY24}.
For $z=\sum_r e_r\otimes x_r$, where $\{e_r\}$ is an orthonormal basis of
$\mathcal E$, use the catalyst $\sum_r e_r\otimes v_r$, where $v_r$ is the
catalyst for $S(O)$ on input $x_r$.  Then
\begin{equation}
W(I_{\mathcal E}\otimes S,O,z)=\sum_r W(S,O,x_r).
\label{eq:transducer-tensor-identity}
\end{equation}

Functional composition replaces a transducer's input oracle by the
transduction action of another transducer.
\begin{proposition}[Functional composition~{\cite[Proposition~9.11]{BJY24}}]
\label{prop:functional-composition}
Let $S_A(O')$ and $S_B(O)$ be canonical transducers.  Suppose that $O'$ is
the only input oracle of $S_A$ and acts on $\mathcal M$, and that $S_B$ has
public space $\mathcal M$ with
\[
S_B(O)\DownTransduce_{\mathcal M}=O'.
\]
Let $\mathcal H_A$ and $\mathcal L_A$ be the public and private spaces of
$S_A$, and suppose its oracle call acts as
$I_{\mathcal E}\otimes O'$ on
$\mathcal E\otimes\mathcal M\subseteq\mathcal L_A$.  Then the
functional composition $S_A\circ S_B$ is a canonical transducer with input
oracle $O$ and
\begin{equation}
(S_A\circ S_B)(O)\DownTransduce_{\mathcal H_A}
=S_A(O')\DownTransduce_{\mathcal H_A}.
\label{eq:functional-composition-action}
\end{equation}
If, for every $z\in\mathcal E\otimes\mathcal M$,
\[
W(I_{\mathcal E}\otimes S_B,O,z)\le c\norm{z}^2,
\]
then, for every public input $\xi$,
\begin{equation}
W(S_A\circ S_B,O,\xi)
\le(1+c)W(S_A,O',\xi).
\label{eq:functional-composition-bound}
\end{equation}
\end{proposition}

\begin{proof}
Write $S_A(O')=S_A^\circ\widetilde O'$ and
$S_B(O)=S_B^\circ\widetilde O_B$ in canonical form, and let
$\mathcal L_B$ be the private space of $S_B$.
The composition acts on
$\mathcal H_A\oplus\mathcal L_A\oplus(\mathcal E\otimes\mathcal L_B)$,
with public space $\mathcal H_A$.
Let $v$ be a catalyst for $S_A(O')$ on input $\xi$, and let $z$ be its
component in $\mathcal E\otimes\mathcal M$, so
$v-z\in\mathcal L_A\cap(\mathcal E\otimes\mathcal M)^\perp$.
Choose a catalyst $w$ for
$I_{\mathcal E}\otimes S_B(O)$ on input $z$.  The transducer
$S_A\circ S_B$ first applies $I_{\mathcal E}\otimes\widetilde O_B$ on
$\mathcal E\otimes(\mathcal M\oplus\mathcal L_B)$.  Its work unitary
then applies $I_{\mathcal E}\otimes S_B^\circ$ on the same space,
followed by $S_A^\circ$ on $\mathcal H_A\oplus\mathcal L_A$.
Each operation acts as the identity on the remaining summands.
Since
\[
(I_{\mathcal E}\otimes S_B(O))(z\oplus w)
=(I_{\mathcal E}\otimes O')z\oplus w,
\]
we have
\[
\begin{aligned}
(S_A\circ S_B)(O)(\xi\oplus v\oplus w)
&=(S_A^\circ\oplus I_{\mathcal E\otimes\mathcal L_B})
\bigl(\xi\oplus(v-z)\oplus
       (I_{\mathcal E}\otimes S_B(O))(z\oplus w)\bigr)\\
&=(S_A^\circ\oplus I_{\mathcal E\otimes\mathcal L_B})
\bigl(\xi\oplus(v-z)\oplus(I_{\mathcal E}\otimes O')z\oplus w\bigr)\\
&=(S_A^\circ\oplus I_{\mathcal E\otimes\mathcal L_B})
\bigl(\widetilde O'(\xi\oplus v)\oplus w\bigr)\\
&=S_A(O')(\xi\oplus v)\oplus w.
\end{aligned}
\]
Thus
$S_A\circ S_B$ has the transduction action in
\eqref{eq:functional-composition-action}.  Its catalyst is $v\oplus w$;
the two components lie in orthogonal spaces, so
\[
W(S_A\circ S_B,O,\xi)=\norm{v}^2+\norm{w}^2
\le(1+c)\norm{v}^2,
\]
because $\norm{z}\le\norm{v}$.  This proves
\eqref{eq:functional-composition-bound}.
\end{proof}

\subsection{Purifier}

Purifiers were introduced by Belovs~\cite{Belovs15} and later implemented
as transducers~\cite{BJY24}.  We use the space-efficient construction of
Belovs and Jeffery~\cite{BJ26}, which takes a reflection about an input
state as its oracle and transduces that state to itself or its negative,
according to whether measuring one qubit yields outcome $1$ with probability
below or above $1/2$.  Let $\mathsf A$ denote that qubit, with Hilbert space
$\mathcal A\cong\mathbb C^2$, and let $\mathcal Y$ be the Hilbert space of the
remaining registers.  Define
\[
P_b:=\proj{b}_{\mathsf A}\otimes I_{\mathcal Y},
\qquad b\in\{0,1\},
\]
and, for a normalized $\varphi\in\mathcal A\otimes\mathcal Y$,
\[
p(\varphi):=\norm{P_1\varphi}^2,
\qquad
\mathcal K(\varphi):=\operatorname{span}\{P_0\varphi,P_1\varphi\}.
\]
The purifier uses an input oracle $O$ that reflects about $\varphi$ on
$\mathcal K(\varphi)$.  In Section~\ref{sec:ground-reflection}, another
transducer supplies this reflection.  To apply
Proposition~\ref{prop:functional-composition}, we put the purifier in
canonical form following~\cite[Proposition~10.4]{BJY24}.  The construction uses a
counter register $\mathsf J$ and a three-state clock register $\mathsf T$,
with
\[
\mathcal J:=\ell^2(\mathbb Z),
\qquad
\mathcal T:=\operatorname{span}\{\ket{-1}_{\mathsf T},\ket0_{\mathsf T},
\ket1_{\mathsf T}\},
\qquad
\mathcal H_{\rm pur}:=
\operatorname{span}\{\ket{-1}_{\mathsf T}\otimes\ket0_{\mathsf J}\}
\otimes\mathcal A\otimes\mathcal Y.
\]
Let $\mathcal L_{\rm pur}$ be the orthogonal complement of
$\mathcal H_{\rm pur}$ in
$\mathcal T\otimes\mathcal J\otimes\mathcal A\otimes\mathcal Y$.
We omit the fixed clock and counter states from public inputs and outputs.
When composing with a circuit on $\mathcal A\otimes\mathcal Y$, that
circuit leaves these two registers fixed.

\begin{theorem}[Purifier in canonical form]
\label{thm:purifier}
Let $O$ be a unitary on $\mathcal A\otimes\mathcal Y$.  There is a canonical
transducer $S_{\rm pur}(O)$ on
$\mathcal T\otimes\mathcal J\otimes\mathcal A\otimes\mathcal Y$ with input
oracle $O$, public space $\mathcal H_{\rm pur}$, and private space
$\mathcal L_{\rm pur}$.  For every
normalized $\varphi$ with $p(\varphi)\ne1/2$, if
\[
O|_{\mathcal K(\varphi)}
=2\varphi\varphi^\dagger-I_{\mathcal K(\varphi)},
\]
then
\[
\varphi\transduce{S_{\rm pur}(O)}
\begin{cases}
+\varphi,&p(\varphi)<1/2,\\
-\varphi,&p(\varphi)>1/2,
\end{cases}
\]
and
\[
W(S_{\rm pur},O,\varphi)
=\cO\!\left(\frac1{|p(\varphi)-1/2|}\right).
\]
\end{theorem}

\begin{proof}
We first write the two-query purifier of
\cite[Eq.~(3.8) and Figure~3.2]{BJ26}, then put it in canonical form.  Put
\[
\mathcal X:=\mathcal J\otimes\mathcal A\otimes\mathcal Y.
\]
The two-query purifier has public space
$\operatorname{span}\{\ket0_{\mathsf J}\}\otimes
\mathcal A\otimes\mathcal Y$, and its private space is
$\operatorname{span}\{\ket j_{\mathsf J}:j\ne0\}\otimes
\mathcal A\otimes\mathcal Y$.  Define the counter shift by
$\operatorname{INC}_{\mathsf J}\ket{j}_{\mathsf J}
=\ket{j+1}_{\mathsf J}$.  Since $j$ ranges over $\mathbb Z$, this shift is
unitary.  Set
\begin{align*}
\widehat O
&:=\proj{0}_{\mathsf J}\otimes I
 +(I-\proj{0}_{\mathsf J})\otimes O,\\
\operatorname{INC}_b
&:=\operatorname{INC}_{\mathsf J}\otimes P_b
 +I_{\mathsf J}\otimes(I-P_b),
 \qquad b\in\{0,1\}.
\end{align*}
Here $\widehat O$ applies $O$ when $\mathsf J\ne0$, while
$\operatorname{INC}_b$ increments $\mathsf J$ when $\mathsf A=b$.  Set
\[
B_0:=\operatorname{INC}_0,
\qquad
B_1:=\operatorname{INC}_1\operatorname{INC}_0^\dagger,
\qquad
B_2:=\operatorname{INC}_1^\dagger(2\proj{0}_{\mathsf J}-I).
\]
On the public and private spaces defined above, the purifier is the
transducer
\begin{equation}
\widehat S_{\rm pur}(O)
:=B_2\widehat O B_1\widehat O B_0.
\label{eq:raw-purifier-factorization}
\end{equation}
It makes two calls to $O$ through $\widehat O$.

For the input $\varphi$ in the theorem, set
$\xi:=\ket0_{\mathsf J}\otimes\varphi$.  The projectors $P_0,P_1$ and the
oracle $O$ preserve
$\mathcal K(\varphi)$, so every operator in
\eqref{eq:raw-purifier-factorization} also preserves
$\mathcal J\otimes\mathcal K(\varphi)$.  On $\mathcal K(\varphi)$, $O$ is
the reflection about $\varphi$.  Theorem~3.7 of~\cite{BJ26} therefore
provides a catalyst
\begin{equation}
v\in
\operatorname{span}\{\ket{j}_{\mathsf J}:j\ne0\}
\otimes\mathcal K(\varphi).
\label{eq:raw-purifier-catalyst-space}
\end{equation}
With
\[
\tau:=
\begin{cases}
+\xi,&p(\varphi)<1/2,\\
-\xi,&p(\varphi)>1/2,
\end{cases}
\]
the cited theorem states that
\begin{equation}
\widehat S_{\rm pur}(O)(\xi\oplus v)=\tau\oplus v,
\qquad
\norm{v}^2
=\cO\!\left(\frac1{|p(\varphi)-1/2|}\right).
\label{eq:raw-purifier-transduction}
\end{equation}

The transducer $\widehat S_{\rm pur}(O)$ is not in the canonical form of
Definition~\ref{def:canonical-transducer}, because the oracle occurs twice
in~\eqref{eq:raw-purifier-factorization}.  To put this two-query transducer in
canonical form, use the three-state clock $\mathsf T$ defined above.
The public space has clock value $-1$, so the oracle can be called
at clock values $0$ and $1$.
Following the circuit-to-transducer conversion of
\cite[Theorem~10.3 and Proposition~10.4]{BJY24}, define
\begin{equation}
\mathcal E:=\operatorname{span}\{\ket0_{\mathsf T},\ket1_{\mathsf T}\}
\otimes\operatorname{span}\{\ket{j}_{\mathsf J}:j\ne0\}.
\label{eq:purifier-oracle-label-space}
\end{equation}
On $\mathcal T\otimes\mathcal X$, let
\[
\widetilde O
:=\proj{-1}_{\mathsf T}\otimes I_{\mathcal X}
 +(\proj{0}_{\mathsf T}+\proj{1}_{\mathsf T})\otimes\widehat O.
\]
It acts as $I_{\mathcal E}\otimes O$ on the private subspace
$\mathcal E\otimes\mathcal A\otimes\mathcal Y$ and as the identity on its
orthogonal complement.

Define the oracle-independent unitary
\begin{equation}
S_{\rm pur}^\circ
:=\bigl(\ket0\bra{-1}\bigr)_{\mathsf T}\otimes B_0
 +\bigl(\ket1\bra0\bigr)_{\mathsf T}\otimes B_1
 +\bigl(\ket{-1}\bra1\bigr)_{\mathsf T}\otimes B_2.
\label{eq:purifier-work-unitary}
\end{equation}
It is unitary because it applies a unitary for each clock value and then
cyclically permutes the three values.  Hence
\begin{equation}
S_{\rm pur}(O):=S_{\rm pur}^\circ\widetilde O
\label{eq:canonical-purifier}
\end{equation}
has the canonical form~\eqref{eq:canonical-form}.

We verify that the clock construction preserves the transduction in
\eqref{eq:raw-purifier-transduction}.  When
$\widehat S_{\rm pur}(O)$ acts on $x:=\xi\oplus v$, its two calls to
$\widehat O$ act on
\[
x_0:=B_0x,
\qquad
x_1:=B_1\widehat O x_0,
\]
respectively.  Define
\begin{equation}
\widetilde v
:=\ket{-1}_{\mathsf T}\otimes v
 +\ket0_{\mathsf T}\otimes x_0
 +\ket1_{\mathsf T}\otimes x_1.
\label{eq:purifier-clock-catalyst}
\end{equation}
Since $v$ is supported on $j\ne0$ by
\eqref{eq:raw-purifier-catalyst-space},
$\ket{-1}_{\mathsf T}\otimes v$ lies in the private space.
The other two terms have clock values $0$ and $1$, so they also lie in the
private space.  Thus $\widetilde v\in\mathcal L_{\rm pur}$.  Using the work unitary
defined in~\eqref{eq:purifier-work-unitary}, we compute
\begin{align*}
S_{\rm pur}(O)
\bigl((\ket{-1}_{\mathsf T}\otimes\xi)\oplus\widetilde v\bigr)
&=S_{\rm pur}^\circ
\bigl(
\ket{-1}_{\mathsf T}\otimes x
+\ket0_{\mathsf T}\otimes\widehat O x_0
+\ket1_{\mathsf T}\otimes\widehat O x_1
\bigr)\\
&=\ket0_{\mathsf T}\otimes B_0x
+\ket1_{\mathsf T}\otimes B_1\widehat O x_0
+\ket{-1}_{\mathsf T}\otimes B_2\widehat O x_1\\
&=\ket0_{\mathsf T}\otimes x_0
+\ket1_{\mathsf T}\otimes x_1
+\ket{-1}_{\mathsf T}\otimes\widehat S_{\rm pur}(O)x\\
&=(\ket{-1}_{\mathsf T}\otimes\tau)\oplus\widetilde v.
\end{align*}
The third line uses the definitions of $x_0,x_1$ and
\eqref{eq:raw-purifier-factorization}; the last uses
\eqref{eq:raw-purifier-transduction}.  Hence $\widetilde v$ is a catalyst,
and its norm satisfies
\begin{align*}
W(S_{\rm pur},O,\varphi)
&=\norm{\widetilde v}^2\\
&=\norm{v}^2+\norm{x_0}^2+\norm{x_1}^2\\
&=2+3\norm{v}^2
=\cO\!\left(\frac1{|p(\varphi)-1/2|}\right).
\end{align*}
Here the third equality uses the unitarity of $B_0$, $B_1$, and $\widehat O$,
the orthogonality of $\xi$ and $v$, and $\norm{\xi}=1$.  The final bound
follows from~\eqref{eq:raw-purifier-transduction} and
$|p(\varphi)-1/2|\le1/2$.
\end{proof}

\paragraph{Finite counter register.}

The construction above uses an infinite counter.  For a power of two $D\ge4$,
replace $\mathcal J=\ell^2(\mathbb Z)$ by $\mathbb C^D$ with basis
$\{\ket{j}_{\mathsf J}:0\le j<D\}$, and set
\[
\operatorname{INC}_D\ket{r}_{\mathsf J}
=\ket{(r+1)\bmod D}_{\mathsf J}.
\]
Since $D$ is a power of two, $\operatorname{INC}_D$ is a ripple-carry
increment on $\log_2D$ qubits and uses $\cO(\log D)$ gates.  Reversing
the circuit or adding a fixed number of controls to each gate preserves
this scaling.
Let $S_{\rm pur}^{[D]}$ be the canonical purifier obtained from
\eqref{eq:canonical-purifier} by this replacement.

The following counter bound adapts the finite-counter argument of
\cite[Proposition~4.3]{BJ26} to $S_{\rm pur}$.
\begin{lemma}[Finite-counter purifier]
\label{lem:finite-counter}
Each application of $S_{\rm pur}^\circ$ or $(S_{\rm pur}^\circ)^\dagger$
changes the value of $\mathsf J$ by at most one, and
$T(S_{\rm pur}^{[D]})=\cO(\log D)$.
\end{lemma}

\begin{proof}
In~\eqref{eq:purifier-work-unitary}, $B_0$ and $B_2$ shift $\mathsf J$
by at most one.  The shifts in $B_1$ are controlled by the orthogonal
projectors $P_0$ and $P_1$, so $B_1$ also shifts $\mathsf J$ by at most
one.  The same is true of their inverses.

Write $b=\log_2D$.  The work unitary~\eqref{eq:purifier-work-unitary}
uses $B_0,B_1,B_2$ controlled by the three clock values.  By definition,
$B_0$ and $B_2$ each require one controlled shift of the $b$-qubit counter,
while $B_1$ requires two; $B_2$ also tests whether the counter is zero.
The controlled shifts use $\cO(b)$ gates.  The zero test uses a
$b$-controlled NOT, which takes $\cO(b)$ one- and two-qubit gates with
one work qubit returned to its initial state~\cite[Corollary~7.4]{BBC95}.
Conditioning on the clock adds a constant number of
controls to each gate, and updating the clock costs $\cO(1)$ gates.  Hence
$T(S_{\rm pur}^{[D]})=\cO(b)=\cO(\log D)$.
\end{proof}

\subsection{Trigonometric polynomial inequalities}
\label{subsec:trigonometric-inequalities}

The query lower bounds use the following growth inequality for real
trigonometric polynomials.
The usual form with $s\le\pi/2$ appears in~\cite[Theorem~5.1.2]{BE95};
the larger range below follows from~\cite[Theorem~1.1]{TY18}.

\begin{lemma}[Trigonometric growth inequality]
\label{lem:trig-growth}
Let
\[
q(\theta)=\sum_{k=-n}^n a_ke^{ik\theta}
\]
be real-valued, and let $|\cdot|$ denote Lebesgue measure.  If, for some
$0<s\le\pi$,
\[
\bigl|\{\theta\in[-\pi,\pi):|q(\theta)|\le1\}\bigr|\ge2\pi-s,
\]
then
\[
\sup_{\theta\in\mathbb R}|q(\theta)|\le\exp(4ns).
\]
\end{lemma}

\begin{proof}
Theorem~1.1 of~\cite{TY18} bounds the supremum by
$\cosh(2n\log(\sec(s/4)+\tan(s/4)))$.
Since $s\le\pi$,
\[
\log(\sec(s/4)+\tan(s/4))
=\int_0^{s/4}\sec u\,du\le\frac{s}{2\sqrt2}.
\]
The supremum is therefore at most $\exp(ns/\sqrt2)\le\exp(4ns)$.
\end{proof}

\subsection{Hybrid argument}

We use the following hybrid argument to compare circuits that differ only in
their oracle gates.
\begin{lemma}[BBBV hybrid argument; cf.~{\cite[Theorem~3.3]{BBBV97}}]
\label{lem:query-hybrid}
Consider circuits indexed by $t=0,1,\ldots,N$ with the same initial state
and inter-query unitaries, but oracle gates $O_{q,t}$ at query position $q$.
Let $\ket{\Psi^t}$ be their final states and $\ket{\Phi_q^0}$ the state
before query $q$ in the reference circuit $t=0$.  Then
\[
\left(\sum_{t=1}^N\norm{\ket{\Psi^t}-\ket{\Psi^0}}^2\right)^{1/2}
\le\sum_q
\left(\sum_{t=1}^N\norm{(O_{q,t}-O_{q,0})\ket{\Phi_q^0}}^2\right)^{1/2}.
\]
In particular, if every query of type $j$ satisfies
$\sum_t\norm{(O_{q,t}-O_{q,0})\ket\Phi}^2\le c_j^2$
for every normalized $\ket\Phi$, then the left-hand side is at most
$\sum_j c_jQ_j$, where $Q_j$ is the number of queries of type $j$.
\end{lemma}

\begin{proof}
For each $t$, replace the reference oracle gates by $O_{q,t}$ from last
to first.  When replacing the gate at query position $q$, all preceding
gates still belong to the reference circuit, so the input is
$\ket{\Phi_q^0}$; all subsequent gates are unitary.  Thus
\[
\norm{\ket{\Psi^t}-\ket{\Psi^0}}
\le\sum_q\norm{(O_{q,t}-O_{q,0})\ket{\Phi_q^0}}.
\]
Taking the Euclidean norm over $t$ and using the triangle inequality gives
\[
\begin{aligned}
\left(\sum_{t=1}^N\norm{\ket{\Psi^t}-\ket{\Psi^0}}^2\right)^{1/2}
&\le\left(\sum_{t=1}^N
  \left[\sum_q\norm{(O_{q,t}-O_{q,0})\ket{\Phi_q^0}}\right]^2\right)^{1/2}\\
&\le\sum_q
  \left(\sum_{t=1}^N\norm{(O_{q,t}-O_{q,0})\ket{\Phi_q^0}}^2\right)^{1/2}.
\end{aligned}
\]
Applying the bound $c_j$ to each query of type $j$ gives the second claim.
\end{proof}

\section{Ground-state reflection}
\label{sec:ground-reflection}

We construct a transducer for $R_g=2\proj{\psi_0}-I$ using the purifier of
Theorem~\ref{thm:purifier}.  Let $\mathsf A$ be its answer qubit, and write
$P_b=\proj b_{\mathsf A}\otimes I$ for $b\in\{0,1\}$, where $I$ acts on the
other registers.  We choose a unitary $V$ that maps each
eigenstate $\ket{\psi_k}$, together with zero-initialized ancillas, to a
state $\varphi_k$ whose answer-$1$ probability is below $1/2$ for $k=0$ and
above $1/2$ for $k\ge1$.
If one oracle reflects about each $\varphi_k$ on $\mathcal K(\varphi_k)$,
the purifier transduces $\varphi_0$ into $\varphi_0$ and
$\varphi_k$ into $-\varphi_k$ for $k\ge1$.  Applying $V^\dagger$ then returns
the ancillas to zero and gives the action of $R_g$ on the system.

\subsection{Constant-accuracy spectral test}
\label{subsec:spectral-test}

To construct $V$, take $U:=U_{1/10,\delta}$ and its corresponding polynomial
$f$ from Lemma~\ref{lem:spectral-filter}.  Let $\mathsf R$ be the ancilla register of
$U$.  The outcome $\mathsf R=0$ has probability at least $81/100$ on
$\ket{\psi_0}$ and at most $1/100$ on
each excited eigenstate.  We obtain $V$ by applying $U$ and reversibly
recording in $\mathsf A$ whether $\mathsf R$ is nonzero.

\begin{lemma}[Constant-accuracy spectral test]
\label{lem:spectral-test}
There is a unitary $V$ acting on $\mathsf A$, $\mathsf R$, and the system
such that, for every system state $\ket\psi$,
\begin{equation}
P_0V(\ket0_{\mathsf A\mathsf R}\otimes\ket\psi)
=\ket0_{\mathsf A\mathsf R}\otimes f(H)\ket\psi.
\label{eq:spectral-test-zero-component}
\end{equation}
For each eigenstate $\ket{\psi_k}$, define
\begin{equation}
\varphi_k:=V(\ket0_{\mathsf A\mathsf R}\otimes\ket{\psi_k}),
\qquad
p_k:=\norm{P_1\varphi_k}^2.
\label{eq:spectral-test-action}
\end{equation}
Under the normalized spectral promise~\eqref{eq:normalized-promise},
$p_k=1-|f(E_k)|^2$, and
\begin{equation}
p_0\le\frac{19}{100},
\qquad
p_k\ge\frac{99}{100}\qquad k\ge1,
\label{eq:spectral-test-probabilities}
\end{equation}
so
\[
\min_k|p_k-1/2|\ge\frac{31}{100}.
\]
Each of $V$ and $V^\dagger$ uses
$\cO(1/\delta)$ queries to the block-encoding of $H$ and
$\cO((a_H+1)/\delta)$ one- and two-qubit gates, where $a_H$ is the number of
ancilla qubits in that block-encoding.
\end{lemma}

\begin{proof}
Define
\[
N_{\ne0}\ket a_{\mathsf A}\ket z_{\mathsf R}
:=\ket{a\mathbin\oplus\mathbf 1_{\{z\ne0\}}}_{\mathsf A}
  \ket z_{\mathsf R},
\qquad
V:=(N_{\ne0}\otimes I)(I_{\mathsf A}\otimes U).
\]
On inputs with $\mathsf A$ initialized to $\ket0$, the answer is $0$
exactly when $\mathsf R=0$.  Hence
\[
\begin{aligned}
P_0V(\ket0_{\mathsf A\mathsf R}\otimes\ket\psi)
&=\ket0_{\mathsf A\mathsf R}\otimes
  \bigl[(\bra0_{\mathsf R}\otimes I)U
  (\ket0_{\mathsf R}\otimes\ket\psi)\bigr]\\
&=\ket0_{\mathsf A\mathsf R}\otimes f(H)\ket\psi,
\end{aligned}
\]
where the second equality uses Lemma~\ref{lem:spectral-filter}.

Applying~\eqref{eq:spectral-test-zero-component} to $\ket{\psi_k}$ gives
\[
P_0\varphi_k
=\ket0_{\mathsf A\mathsf R}\otimes f(H)\ket{\psi_k}
=f(E_k)\ket0_{\mathsf A\mathsf R}\ket{\psi_k}.
\]
The state $\varphi_k$ is normalized and $P_0+P_1=I$, so
\[
p_k=1-\norm{P_0\varphi_k}^2=1-|f(E_k)|^2.
\]
Lemma~\ref{lem:spectral-filter} gives $|f(E_0)|\ge9/10$ and
$|f(E_k)|\le1/10$ for $k\ge1$.  Thus
\[
p_0\le1-(9/10)^2=19/100,
\qquad
p_k\ge1-(1/10)^2=99/100\quad(k\ge1).
\]
The smaller margin from $1/2$ is $1/2-19/100=31/100$.

Since $\eta=1/10$ is fixed, Lemma~\ref{lem:spectral-filter} shows that $U$ uses
$\cO(1/\delta)$ Hamiltonian queries and
$\cO((a_H+1)/\delta)$ one- and two-qubit gates.  Let $r$ be the number of
qubits in $\mathsf R$, so $r=a_H+\cO(1)$.  To implement $N_{\ne0}$, apply
$X$ to $\mathsf A$, then an $r$-controlled NOT targeting $\mathsf A$,
conjugating each control in $\mathsf R$ by $X$.  This controlled NOT can be
implemented with $\cO(r)$ one- and two-qubit gates, so $N_{\ne0}$ adds
$\cO(a_H+1)$ gates, absorbed in
$\cO((a_H+1)/\delta)$ because $\delta<1$.  Taking adjoints preserves both
bounds, which proves the stated costs for $V$ and $V^\dagger$.
\end{proof}

\subsection{Reflection about the test states}
\label{subsec:orthogonal-invariant-subspaces}

To apply Theorem~\ref{thm:purifier} to every $\varphi_k$, we need one oracle
whose restriction to each $\mathcal K(\varphi_k)$ is the reflection about
$\varphi_k$.  Define
\begin{equation}
\Pi:=V(\proj{0}_{\mathsf A\mathsf R}\otimes I)V^\dagger,
\qquad
O_{\rm ref}:=2\Pi-I
=V\bigl(2(\proj0_{\mathsf A\mathsf R}\otimes I)-I\bigr)V^\dagger.
\label{eq:spectral-test-subspaces}
\end{equation}

\begin{lemma}[Reflection oracle]
\label{lem:spectral-test-subspaces}
The spaces $\mathcal K(\varphi_k)$ are mutually orthogonal.  Each
$\mathcal K(\varphi_k)$ is
invariant under $P_0$, $P_1$, $\Pi$, and $O_{\rm ref}$, and
\begin{equation}
\Pi|_{\mathcal K(\varphi_k)}=\varphi_k\varphi_k^\dagger,
\qquad
O_{\rm ref}|_{\mathcal K(\varphi_k)}
=2\varphi_k\varphi_k^\dagger-I_{\mathcal K(\varphi_k)}.
\label{eq:reflecting-oracle-on-subspace}
\end{equation}
\end{lemma}

\begin{proof}
To show that the spaces $\mathcal K(\varphi_k)$ are orthogonal, we compare
their spanning vectors $P_0\varphi_k$ and $P_1\varphi_k$.
The vectors $\varphi_k$ are orthonormal because $V$ maps the orthonormal
inputs $\ket0_{\mathsf A\mathsf R}\ket{\psi_k}$ to them.  Using
$P_0^2=P_0$ and the $P_0$ action in
\eqref{eq:spectral-test-zero-component}, we compute, for every $k$
and $\ell$,
\begin{equation}
\begin{aligned}
\varphi_k^\dagger P_0\varphi_\ell
&=(P_0\varphi_k)^\dagger(P_0\varphi_\ell)\\
&=\bigl(\ket0_{\mathsf A\mathsf R}\otimes f(H)\ket{\psi_k}\bigr)^\dagger
  \bigl(\ket0_{\mathsf A\mathsf R}\otimes f(H)\ket{\psi_\ell}\bigr)\\
&=f(E_k)^*f(E_\ell)\langle\psi_k|\psi_\ell\rangle\\
&=|f(E_k)|^2\delta_{k\ell}
=(1-p_k)\delta_{k\ell},\\
\varphi_k^\dagger P_1\varphi_\ell
&=\varphi_k^\dagger\varphi_\ell
  -\varphi_k^\dagger P_0\varphi_\ell
=\delta_{k\ell} -(1-p_k)\delta_{k\ell}=p_k\delta_{k\ell}.
\end{aligned}
\label{eq:answer-component-inner-products}
\end{equation}
For $k\ne\ell$,
\[
(P_b\varphi_k)^\dagger(P_{b'}\varphi_\ell)
=\delta_{bb'}\varphi_k^\dagger P_b\varphi_\ell=0,
\qquad b,b'\in\{0,1\}.
\]
Since $P_0\varphi_k$ and $P_1\varphi_k$ span $\mathcal K(\varphi_k)$, this proves
$\mathcal K(\varphi_k)\perp\mathcal K(\varphi_\ell)$.

Using $I=\sum_\ell\proj{\psi_\ell}$ in the definition of $\Pi$,
\[
\begin{aligned}
\Pi
&=V\left(\proj0_{\mathsf A\mathsf R}\otimes
  \sum_\ell\proj{\psi_\ell}\right)V^\dagger\\
&=\sum_\ell
  V(\ket0_{\mathsf A\mathsf R}\ket{\psi_\ell})
  (\bra0_{\mathsf A\mathsf R}\bra{\psi_\ell})V^\dagger\\
&=\sum_\ell\varphi_\ell\varphi_\ell^\dagger.
\end{aligned}
\]
For $x\in\mathcal K(\varphi_k)$, the relation
$\mathcal K(\varphi_k)\perp\mathcal K(\varphi_\ell)$ gives
$\varphi_\ell^\dagger x=0$ when $\ell\ne k$, so
\[
\Pi x
=\sum_\ell\varphi_\ell(\varphi_\ell^\dagger x)
=\varphi_k(\varphi_k^\dagger x).
\]
Since $\varphi_k=P_0\varphi_k+P_1\varphi_k$ belongs to $\mathcal K(\varphi_k)$,
$\Pi x\in\mathcal K(\varphi_k)$ for every $x\in\mathcal K(\varphi_k)$, and
$\Pi|_{\mathcal K(\varphi_k)}=\varphi_k\varphi_k^\dagger$.  Moreover,
\[
P_1(P_0\varphi_k)=0,
\qquad
P_1(P_1\varphi_k)=P_1\varphi_k,
\]
so $P_1$ preserves $\mathcal K(\varphi_k)$.  Since $P_0=I-P_1$, the projector $P_0$
also preserves $\mathcal K(\varphi_k)$.  Finally,
$O_{\rm ref}=2\Pi-I$ preserves $\mathcal K(\varphi_k)$ and satisfies
\[
O_{\rm ref}|_{\mathcal K(\varphi_k)}
=2\varphi_k\varphi_k^\dagger-I_{\mathcal K(\varphi_k)}.
\]
\end{proof}

\subsection{Purifier on superpositions}
\label{subsec:coherent-purifier}

The preceding lemma gives one reflection oracle $O_{\rm ref}$ for all
$\varphi_k$.  We use it in the canonical transducer $S_{\rm pur}(O_{\rm ref})$
of Theorem~\ref{thm:purifier}, with $\mathcal Y$ the joint Hilbert space of
$\mathsf R\mathsf S$.  Its public states have $\mathsf T=-1$ and $\mathsf J=0$
by the definition of $\mathcal H_{\rm pur}$.  We must also check that the
transduction complexity remains constant on superpositions of the
$\varphi_k$.

\begin{lemma}[Purifier on superpositions]
\label{lem:purifier-direct-sum}
For every normalized vector $\sum_k a_k\varphi_k$,
\begin{equation}
\sum_k a_k\varphi_k
\transduce{S_{\rm pur}(O_{\rm ref})}
a_0\varphi_0-\sum_{k\ge1}a_k\varphi_k.
\label{eq:coherent-purifier-action}
\end{equation}
Moreover,
\begin{equation}
W\!\left(S_{\rm pur},O_{\rm ref},
\sum_k a_k\varphi_k\right)
=\cO(1).
\end{equation}
\end{lemma}

\begin{proof}
By Lemmas~\ref{lem:spectral-test-subspaces} and~\ref{lem:spectral-test},
\[
O_{\rm ref}|_{\mathcal K(\varphi_k)}
=2\varphi_k\varphi_k^\dagger-I_{\mathcal K(\varphi_k)},
\qquad
p_0<\frac12<p_k\quad(k\ge1).
\]
The reflection identity and probability bounds above allow us to apply
Theorem~\ref{thm:purifier} to each $\varphi_k$ with the same oracle
$O_{\rm ref}$.  To combine the resulting
transductions, we choose their catalysts for $S_{\rm pur}(O_{\rm ref})$ in
orthogonal spaces.

Fix $k$.  In the proof of Theorem~\ref{thm:purifier}, the input
$\xi=\ket0_{\mathsf J}\otimes\varphi_k$ and the catalyst $v$ of
$\widehat S_{\rm pur}(O_{\rm ref})$ both lie in
$\mathcal J\otimes\mathcal K(\varphi_k)$, by
\eqref{eq:raw-purifier-catalyst-space}.  Lemma~\ref{lem:spectral-test-subspaces}
shows that $P_0$, $P_1$, and $O_{\rm ref}$ preserve
$\mathcal K(\varphi_k)$.  Consequently $B_0$, $B_1$, and $\widehat O$
preserve $\mathcal J\otimes\mathcal K(\varphi_k)$, so the vectors
$x_0=B_0(\xi\oplus v)$ and $x_1=B_1\widehat O x_0$ lie there as well.
The clock construction~\eqref{eq:purifier-clock-catalyst} therefore gives a
catalyst for $S_{\rm pur}(O_{\rm ref})$ on input $\varphi_k$ of the form
\[
v_k=\ket{-1}_{\mathsf T}\otimes v
    +\ket0_{\mathsf T}\otimes x_0
    +\ket1_{\mathsf T}\otimes x_1
\in\mathcal T\otimes\mathcal J\otimes\mathcal K(\varphi_k).
\]
Substituting $p_0<1/2<p_k$ into the action of
Theorem~\ref{thm:purifier}, with the catalyst $v_k$ above, gives
\[
S_{\rm pur}(O_{\rm ref})
\bigl(\varphi_k\oplus v_k\bigr)
=
\begin{cases}
\varphi_k\oplus v_k,&k=0,\\
-\varphi_k\oplus v_k,&k\ge1,
\end{cases}
\]
and $\norm{v_k}^2=\cO(1/|p_k-1/2|)$.

For $k\ne\ell$, Lemma~\ref{lem:spectral-test-subspaces} gives
$\mathcal K(\varphi_k)\perp\mathcal K(\varphi_\ell)$, hence
$v_k\perp v_\ell$.  The oracle
$O_{\rm ref}$ is fixed, so the transductions for all $k$ use the same
unitary $S_{\rm pur}(O_{\rm ref})$.  For input
$\sum_k a_k\varphi_k$, set $v:=\sum_k a_kv_k$.  By linearity,
\begin{align*}
S_{\rm pur}(O_{\rm ref})
\left(\left(\sum_k a_k\varphi_k\right)\oplus v\right)
&=\sum_k a_k S_{\rm pur}(O_{\rm ref})(\varphi_k\oplus v_k)\\
&=\left(a_0\varphi_0-\sum_{k\ge1}a_k\varphi_k\right)\oplus v.
\end{align*}
Thus $v$ is a catalyst for~\eqref{eq:coherent-purifier-action}.  Since
$\varphi_k\in\mathcal K(\varphi_k)$ and these spaces are orthogonal, the
normalization of $\sum_k a_k\varphi_k$ gives $\sum_k|a_k|^2=1$.  The
squared catalyst norm therefore satisfies
\begin{align*}
W\!\left(S_{\rm pur},O_{\rm ref},\sum_k a_k\varphi_k\right)
=\norm{v}^2
&=\sum_k|a_k|^2\norm{v_k}^2\\
&=\cO\!\left(
\sum_k\frac{|a_k|^2}{|p_k-1/2|}
\right)=\cO(1),
\end{align*}
using $|p_k-1/2|\ge31/100$ in the last step.
\end{proof}

\subsection{Ground-state reflection transducer}
\label{subsec:ground-state-reflection}

We now compose the purifier with $V$ and $V^\dagger$ to obtain a canonical
transducer for $R_g$.  Write $O_V:=V\oplus V^\dagger$ so that a label selects $V$ or
$V^\dagger$ in each oracle query.

\begin{proposition}[Ground-state reflection]
\label{prop:ground-reflection}
There is a canonical transducer $S_g(O_V)$ such that, for every normalized
system state $\ket{\psi}$,
\begin{equation}
\ket0_{\mathsf A\mathsf R}\ket{\psi}
\transduce{S_g(O_V)}
\ket0_{\mathsf A\mathsf R}R_g\ket{\psi},
\qquad
W(S_g,O_V,\ket0_{\mathsf A\mathsf R}\ket{\psi})=\cO(1).
\label{eq:ground-reflection-complexity}
\end{equation}
\end{proposition}

\begin{proof}
We first express $R_g$ using the purifier's transduction action, then
construct $S_g$ by functional and sequential composition.
Write an arbitrary system state as
$\ket\psi=\sum_k a_k\ket{\psi_k}$, where $\sum_k|a_k|^2=1$.
By linearity of $V$ and~\eqref{eq:spectral-test-action},
\[
V(\ket0_{\mathsf A\mathsf R}\otimes\ket\psi)
=\sum_k a_kV(\ket0_{\mathsf A\mathsf R}\otimes\ket{\psi_k})
=\sum_k a_k\varphi_k.
\]
By Lemma~\ref{lem:purifier-direct-sum}, $S_{\rm pur}(O_{\rm ref})$ transduces
$\sum_k a_k\varphi_k$ into
$a_0\varphi_0-\sum_{k\ge1}a_k\varphi_k$.  Applying $V^\dagger$ and using
\eqref{eq:spectral-test-action},
\[
V^\dagger\left(a_0\varphi_0-\sum_{k\ge1}a_k\varphi_k\right)
=\ket0_{\mathsf A\mathsf R}
 \left(a_0\ket{\psi_0}-\sum_{k\ge1}a_k\ket{\psi_k}\right)
=\ket0_{\mathsf A\mathsf R}R_g\ket\psi.
\]
Together, these three steps give
\begin{equation}
\ket0_{\mathsf A\mathsf R}\ket\psi
\xrightarrow{V}\sum_k a_k\varphi_k
\transduce{S_{\rm pur}(O_{\rm ref})}
a_0\varphi_0-\sum_{k\ge1}a_k\varphi_k
\xrightarrow{V^\dagger}
\ket0_{\mathsf A\mathsf R}R_g\ket\psi.
\label{eq:ground-reflection-clean-action}
\end{equation}

We now construct the canonical transducer $S_g$ for the transformation
in~\eqref{eq:ground-reflection-clean-action} and bound its transduction
complexity.
Let $\mathcal M=\mathcal A\otimes\mathcal Y$ be the space on which
$O_{\rm ref}$ and $V$ act.
By~\eqref{eq:spectral-test-subspaces}, $O_{\rm ref}$ uses one $V$ and one
$V^\dagger$, hence two queries to $O_V$.  Use
Lemma~\ref{lem:circuit-to-transducer} to obtain a canonical transducer
$S_{\rm ref}$ for the circuit implementing $O_{\rm ref}$.  Similarly, let
$S_V$ and $S_{V^\dagger}$ be the transducers for the circuits that apply $V$ and
$V^\dagger$ with one query to $O_V$.  For every $x\in\mathcal M$, the conversion
gives
\begin{align*}
W(S_{\rm ref},O_V,x)&\le2\norm{x}^2,\\
W(S_V,O_V,x)&\le\norm{x}^2,\\
W(S_{V^\dagger},O_V,x)&\le\norm{x}^2.
\end{align*}
Define
\begin{equation}
S_g:=S_{V^\dagger}*(S_{\rm pur}\circ S_{\rm ref})*S_V.
\label{eq:ground-reflection-transducer}
\end{equation}
The transduction actions of $S_{\rm ref}(O_V)$, $S_V$, and $S_{V^\dagger}$
are $O_{\rm ref}$, $V$, and $V^\dagger$, respectively.  Functional
composition replaces the $O_{\rm ref}$ call in $S_{\rm pur}$ with
$S_{\rm ref}$, while sequential composition combines the resulting action
with $V$ and $V^\dagger$.  Propositions~\ref{prop:functional-composition}
and~\ref{prop:sequential-composition} show that $S_g$ is canonical and
has the action in~\eqref{eq:ground-reflection-clean-action}.  Thus it transduces
$\ket0_{\mathsf A\mathsf R}\ket\psi$ into
$\ket0_{\mathsf A\mathsf R}R_g\ket\psi$.

It remains to bound the transduction complexity of $S_g$ on
$\ket0_{\mathsf A\mathsf R}\ket\psi$.  We first bound the middle factor
$S_{\rm pur}\circ S_{\rm ref}$.  Lemma~\ref{lem:purifier-direct-sum} gives
\[
W\!\left(S_{\rm pur},O_{\rm ref},\sum_k a_k\varphi_k\right)
=\cO(1).
\]
By~\eqref{eq:purifier-oracle-label-space}, the purifier calls
$O_{\rm ref}$ on $\mathcal E\otimes\mathcal M$, leaving $\mathcal E$
unchanged.  The bound $W(S_{\rm ref},O_V,x)\le2\norm{x}^2$ and
\eqref{eq:transducer-tensor-identity} imply, for every
$z\in\mathcal E\otimes\mathcal M$,
\[
W(I_{\mathcal E}\otimes S_{\rm ref},O_V,z)
\le2\norm{z}^2.
\]
Proposition~\ref{prop:functional-composition} therefore yields
\[
W\!\left(S_{\rm pur}\circ S_{\rm ref},O_V,
\sum_k a_k\varphi_k\right)
\le3W\!\left(S_{\rm pur},O_{\rm ref},\sum_k a_k\varphi_k\right)
=\cO(1).
\]
The two intermediate public states in the sequential composition defining
$S_g$ in~\eqref{eq:ground-reflection-transducer} are
$\sum_k a_k\varphi_k$ and $a_0\varphi_0-\sum_{k\ge1}a_k\varphi_k$.
Both are normalized because the $\varphi_k$ are orthonormal.
By~\eqref{eq:sequential-composition}, their squared norms add $1+1$ to
the transduction complexity, while the three component transduction
complexities are bounded by $1$, $\cO(1)$, and $1$, respectively.  Hence
\[
W(S_g,O_V,\ket0_{\mathsf A\mathsf R}\ket\psi)
\le(1+1)+(1+\cO(1)+1)
=\cO(1),
\]
which proves~\eqref{eq:ground-reflection-complexity}.
\end{proof}

\section{Ground-state overlap amplification}
\label{sec:constant-precision-amplification}

We now use the ground-state reflection transducer from the preceding
section for amplitude amplification.

\subsection{Randomized Grover iterations}
\label{subsec:ideal-randomized-amplification}

Amplitude amplification uses $R_g$ and the reflection about
$\ket\phi$~\cite{BHMT02}.
Set
\begin{equation}
R_\phi:=2\proj\phi-I
=U_I(2\proj0-I)U_I^\dagger,
\qquad
G:=R_\phi R_g.
\label{eq:ideal-grover-iterate}
\end{equation}
The number of iterations needed to amplify the ground-state component depends
on the initial overlap
$\beta=|\langle\psi_0|\phi\rangle|$, but the algorithm knows only
$\beta\ge\gamma$.
We therefore choose the count at random from a range of
size $\cO(1/\gamma)$~\cite{BBHT98,BHMT02}.

To analyze the randomized procedure, we first calculate $G^j\ket\phi$.
Choose the global phase
of $\ket{\psi_0}$ so that $\langle\psi_0|\phi\rangle=\beta$.  For
$\beta<1$, put
\[
\ket{\phi_\perp}
:=\frac{\ket\phi-\beta\ket{\psi_0}}{\sqrt{1-\beta^2}},
\qquad
\theta=\arcsin\beta.
\]
Then $\ket{\phi_\perp}$ is normalized and orthogonal to $\ket{\psi_0}$, and
\[
\ket\phi
=\sin\theta\ket{\psi_0}+\cos\theta\ket{\phi_\perp}.
\]
Both reflections in~\eqref{eq:ideal-grover-iterate} preserve
$\operatorname{span}\{\ket{\psi_0},\ket{\phi_\perp}\}$.  Multiplying their
matrices in the ordered basis $(\ket{\psi_0},\ket{\phi_\perp})$ gives
\[
R_g=\begin{pmatrix}1&0\\0&-1\end{pmatrix},
\qquad
G=-\begin{pmatrix}
\cos(2\theta)&\sin(2\theta)\\
-\sin(2\theta)&\cos(2\theta)
\end{pmatrix}.
\]
Applying the matrix for $G$ $j$ times to
$(\sin\theta,\cos\theta)^{\mathsf T}$ gives
\begin{equation}
G^j\ket\phi
=(-1)^j\left[
\sin((2j+1)\theta)\ket{\psi_0}
+\cos((2j+1)\theta)\ket{\phi_\perp}
\right].
\label{eq:ideal-grover-rotation}
\end{equation}

Take $M:=\lceil2/\gamma\rceil$.  The ideal procedure outputs $\ket\phi$
with probability $1/2$; otherwise it samples $j$ uniformly from
$\{0,\ldots,M-1\}$ and outputs $G^j\ket\phi$.  Its average output state is
\begin{equation}
\rho_{\mathrm{ideal}}
:=\frac12\proj\phi
+\frac1{2M}\sum_{j=0}^{M-1}G^j\proj{\phi}(G^{\dagger})^{j}.
\label{eq:ideal-amplified-mixture}
\end{equation}
The following lemma shows that $\rho_{\mathrm{ideal}}$ has constant overlap
with $\ket{\psi_0}$ for every $\beta\ge\gamma$.

\begin{lemma}[Average over Grover iterations]
\label{lem:ideal-averaged-grover-state}
The state $\rho_{\mathrm{ideal}}$ satisfies
\[
\bra{\psi_0}\rho_{\mathrm{ideal}}\ket{\psi_0}\ge\frac3{16}.
\]
\end{lemma}

\begin{proof}
If $\beta=1$, then $\rho_{\mathrm{ideal}}=\proj{\psi_0}$.  Suppose $\beta<1$
and use the angle $\theta$ above.  Define
\[
p_M(\theta)
:=\frac1M\sum_{j=0}^{M-1}\sin^2((2j+1)\theta).
\]
Using the rotation~\eqref{eq:ideal-grover-rotation} in the mixture
\eqref{eq:ideal-amplified-mixture}, the ground-state overlap of
$\rho_{\mathrm{ideal}}$ is
\[
\bra{\psi_0}\rho_{\mathrm{ideal}}\ket{\psi_0}
=\frac12\sin^2\theta+\frac12p_M(\theta).
\]
Using $\sin^2x=(1-\cos(2x))/2$,
\begin{align}
p_M(\theta)
&=\frac12-\frac1{2M}\sum_{j=0}^{M-1}
\cos((4j+2)\theta)
\notag\\
&=\frac12-\frac{\sin(4M\theta)}{4M\sin(2\theta)}.
\label{eq:averaged-grover-probability}
\end{align}
The second equality follows from
\[
\sum_{j=0}^{M-1}\cos((4j+2)\theta)
=\Re\!\left(e^{2i\theta}\sum_{j=0}^{M-1}e^{4ij\theta}\right)=\Re\!\left(e^{2i\theta}
  \frac{1-e^{4iM\theta}}{1-e^{4i\theta}}\right)
=\frac{\sin(4M\theta)}{2\sin(2\theta)}.
\]
For $0<\theta\le\pi/4$,
\[
\sin(2\theta)
=2\sin\theta\cos\theta
\ge\sqrt2\,\gamma,
\]
because $\sin\theta=\beta\ge\gamma$ and
$\cos\theta\ge1/\sqrt2$.  Since $M\ge2/\gamma$,
\eqref{eq:averaged-grover-probability} implies
\[
p_M(\theta)
\ge\frac12-\frac1{4M\sin(2\theta)}
\ge\frac12-\frac1{4M\sqrt2\gamma}
\ge\frac12-\frac1{8\sqrt2}
>\frac38.
\]
For $0<\theta\le\pi/4$, this gives
$\bra{\psi_0}\rho_{\mathrm{ideal}}\ket{\psi_0}\ge p_M(\theta)/2>3/16$.

For $\pi/4<\theta<\pi/2$, $p_M(\theta)\ge0$ and $\sin^2\theta>1/2$, so
$\bra{\psi_0}\rho_{\mathrm{ideal}}\ket{\psi_0}\ge\sin^2\theta/2>1/4$.
The two cases prove the lemma.
\end{proof}

\subsection{Transducer implementation}
\label{subsec:implementing-amplification}

To prepare the states $G^j\ket\phi$ in~\eqref{eq:ideal-amplified-mixture},
we compose transducers for $R_g$ and $R_\phi$ and then approximate the
resulting transduction action by a circuit.
Set $O_I:=U_I\oplus U_I^\dagger$, and let $S_\phi(O_I)$ be the canonical
transducer obtained from the circuit for $I_{\mathsf A\mathsf R}\otimes R_\phi$
by Lemma~\ref{lem:circuit-to-transducer}.
Both $S_\phi$ and $S_g$ have the joint Hilbert space of
$\mathsf A\mathsf R\mathsf S$ as their public space.

\begin{proposition}[Transducer for a fixed iteration count]
\label{prop:fixed-branch-transducer}
For each integer $j\ge0$, let $S_0$ be the identity transducer and, for
$j\ge1$, define
\begin{equation}
S_j:=(S_\phi*S_g)^j.
\label{eq:amplification-branch-transducer}
\end{equation}
Here $(S_\phi*S_g)^j$ denotes the sequential composition of $j$ copies of
$S_\phi*S_g$.  Then
\begin{equation}
\ket0_{\mathsf A\mathsf R}\ket\phi
\transduce{S_j(O_V\oplus O_I)}
\ket0_{\mathsf A\mathsf R}G^j\ket\phi,
\label{eq:amplification-branch-action}
\end{equation}
and
\begin{equation}
W(S_j,O_V\oplus O_I,\ket0_{\mathsf A\mathsf R}\ket\phi)=\cO(j).
\label{eq:fixed-branch-costs}
\end{equation}
\end{proposition}

\begin{proof}
The circuit $R_\phi=U_I(2\proj0-I)U_I^\dagger$ acts on $\mathsf S$ and
uses two queries to $O_I$.
Lemma~\ref{lem:circuit-to-transducer} therefore gives $S_\phi$ the
transduction action $I_{\mathsf A\mathsf R}\otimes R_\phi$ and transduction
complexity at most $2$ on every normalized public input.

For $j\ge1$, we first verify the action of $S_j$.  Its public input has
$\mathsf A\mathsf R$ in $\ket0_{\mathsf A\mathsf R}$.  For
$1\le t\le j$, Proposition~\ref{prop:ground-reflection} and the action of
$S_\phi$ give
\[
\ket0_{\mathsf A\mathsf R}G^{t-1}\ket\phi
\transduce{S_g}
\ket0_{\mathsf A\mathsf R}R_gG^{t-1}\ket\phi
\transduce{S_\phi}
\ket0_{\mathsf A\mathsf R}G^t\ket\phi.
\]
Thus every public input to $S_g$ in this transduction sequence has
$\mathsf A\mathsf R$ in $\ket0_{\mathsf A\mathsf R}$, and sequential
composition gives
\eqref{eq:amplification-branch-action}.

We next bound the transduction complexity on
$\ket0_{\mathsf A\mathsf R}\ket\phi$.  For each
$1\le t\le j$, Proposition~\ref{prop:ground-reflection} and
Lemma~\ref{lem:circuit-to-transducer} give
\[
\begin{aligned}
W(S_g,O_V,\ket0_{\mathsf A\mathsf R}G^{t-1}\ket\phi)&=\cO(1),\\
W(S_\phi,O_I,\ket0_{\mathsf A\mathsf R}R_gG^{t-1}\ket\phi)&\le2.
\end{aligned}
\]
All public states in the sequence are normalized.  The first sum in
\eqref{eq:sequential-composition} is therefore $2j-1$; its second sum has
$j$ terms from $S_\phi$ and $j$ from $S_g$.  Hence
\[
W(S_j,O_V\oplus O_I,\ket0_{\mathsf A\mathsf R}\ket\phi)
\le(2j-1)+j\bigl(2+\cO(1)\bigr)=\cO(j),
\]
as claimed in~\eqref{eq:fixed-branch-costs}.
\end{proof}

We now turn the transduction in~\eqref{eq:amplification-branch-action}
into a finite circuit for each $j<M$.
\begin{proposition}[Circuit for a fixed iteration count]
\label{prop:fixed-branch-circuit}
For each $0\le j<M$, there is a circuit $C_j$.  On input
$\ket0_{\mathsf A\mathsf R}\ket\phi$, its reduced output on $\mathsf S$,
denoted $\rho_j$, satisfies
\begin{equation}
\Dtr\!\left(\rho_j,G^j\proj{\phi}(G^{\dagger})^j\right)
\le\frac1{32}.
\label{eq:amplification-branch-error}
\end{equation}
$C_j$ uses $\cO(j/\delta)$ queries to $U_H$ and $U_H^\dagger$ and
$\cO(j)$ queries to $U_I$ and $U_I^\dagger$.  If $\mathsf S$ has $n$
qubits, the number of additional one- and two-qubit gates in $C_j$ is
\begin{equation}
\cO\!\left(
\frac{(a_H+1)j}{\delta}
+j(n+a_H+\log M)
\right).
\label{eq:fixed-branch-gate-cost}
\end{equation}
\end{proposition}

\begin{proof}
\emph{Circuit error.}
For $j\ge1$, the transduction complexity of $S_j$ on
$\ket0_{\mathsf A\mathsf R}\ket\phi$ is $\cO(j)$ by
\eqref{eq:fixed-branch-costs}.  Lemma~\ref{lem:circuit-implementation},
applied with error $1/32$, gives a circuit with $K=\cO(j)$ controlled
executions of $S_j$.  Its output is within Euclidean distance $1/32$
of $\ket0_{\mathsf A\mathsf R}G^j\ket\phi$.  Since pure-state trace
distance is at most Euclidean distance and tracing out ancillas cannot
increase trace distance, the reduced system output satisfies
\eqref{eq:amplification-branch-error}.

\emph{Counter truncation.}
The purifier uses the infinite-dimensional space
$\mathcal J=\ell^2(\mathbb Z)$.  To obtain a finite circuit, we show that
replacing every occurrence of $\mathcal J$ by $\mathbb C^D$, with
$D=\cO(M)$, preserves the circuit output.
We first locate the counter spaces in $S_j=(S_\phi*S_g)^j$.
The circuit-to-transducer constructions for $S_\phi$, $S_V$,
$S_{V^\dagger}$, and $S_{\rm ref}$ introduce no infinite-dimensional
counter.  Thus, in each $S_g$ in~\eqref{eq:ground-reflection-transducer},
only the domain of $S_{\rm pur}\circ S_{\rm ref}$ contains a $\mathcal J$
factor.
Let $\mathcal M$ be the joint Hilbert space of $\mathsf A\mathsf R\mathsf S$,
and write $\mathcal L_{\rm ref}$ for the private space of $S_{\rm ref}$.
By Propositions~\ref{prop:sequential-composition}
and~\ref{prop:functional-composition}, the $j$ occurrences of
$S_{\rm pur}\circ S_{\rm ref}$ act on the orthogonal spaces
\begin{equation}
\bigoplus_{t=1}^j
\left[(\mathcal T\otimes\mathcal J\otimes\mathcal M)
\oplus(\mathcal E\otimes\mathcal L_{\rm ref})\right]_t.
\label{eq:amplification-counter-spaces}
\end{equation}
The subscript $t$ labels the occurrence within the $t$th $S_g$.
The space $\mathcal T\otimes\mathcal J\otimes\mathcal M$
is the domain of $S_{\rm pur}$, while
$\mathcal E\otimes\mathcal L_{\rm ref}$ is the additional private
space introduced by functional composition.  By
\eqref{eq:purifier-oracle-label-space}, $\mathcal E$ has the tensor
factor $\operatorname{span}\{\ket r:r\ne0\}\subseteq\mathcal J$.

We now bound the change in the counter index $r$ during one execution
of $S_j$, for a fixed $t$:
\begin{itemize}
\item The outer sequential composition has work unitary
\begin{equation}
S_j^\circ=P\left(\bigoplus_{t=1}^j
  (S_g^\circ\oplus S_\phi^\circ)\right).
\label{eq:fixed-branch-work-unitary}
\end{equation}
On the space labelled $t$ in~\eqref{eq:amplification-counter-spaces},
the direct sum applies the $t$th $S_g^\circ$ once; its other terms
act on orthogonal spaces.
\item Within $S_g^\circ$, the three-factor sequential composition
in~\eqref{eq:ground-reflection-transducer} applies
$S_V^\circ$, $(S_{\rm pur}\circ S_{\rm ref})^\circ$, and
$S_{V^\dagger}^\circ$ on their respective orthogonal spaces.
Consequently, $(S_{\rm pur}\circ S_{\rm ref})^\circ$ is applied once
on the space labelled $t$.
\item In $(S_{\rm pur}\circ S_{\rm ref})^\circ$, functional composition
applies $I_{\mathcal E}\otimes S_{\rm ref}^\circ$ and then applies
$S_{\rm pur}^\circ$ once
(Proposition~\ref{prop:functional-composition}).  The operator
$I_{\mathcal E}\otimes S_{\rm ref}^\circ$ preserves $r$, since it acts
as the identity on $\mathcal E$.  By Lemma~\ref{lem:finite-counter},
$S_{\rm pur}^\circ$ changes $r$ by at most one.
\end{itemize}
At both levels of sequential composition, the public-space
permutations fix private vectors.  Any intermediate public state with a
$\mathcal J$ factor lies in $\mathcal H_{\rm pur}$ and has $r=0$.
These permutations therefore do not increase $|r|$.  The controls on
oracle calls depend on $r$ only through the test $r=0$ and do not change it.
Hence one
execution of $S_j$ increases $|r|$ by at most one in each counter space.
The implementation circuit starts from the public input with zero
private component, so no counter state with $r\ne0$ is present initially.
The gates other than the controlled executions of $S_j$ act on a separate
control register and leave the purifier counter $\mathsf J$ unchanged
\cite[Proof of Theorem~5.5]{BJY24}.  Hence, during the first $K$
executions, the state has support only on
$\operatorname{span}\{\ket r:|r|\le K\}$ in each $\mathcal J$ factor.

Take
$D=2^{\lceil\log_2(2K+1)\rceil}=\cO(M)$, since $K=\cO(j)$ and $j<M$.
Since $D>2K$, the residues of $-K,\ldots,K$ are distinct, and only
$r=0$ maps to $0$.  The circuit depends on $r$ only through shifts
and $\proj0$, which agree under
$\ket r\mapsto\ket{r\bmod D}$ on the states reached during these $K$
executions.  Replacing $\mathcal J$ by $\mathbb C^D$ throughout the
compositions therefore preserves the
reduced system output and the bound
in~\eqref{eq:amplification-branch-error}.

\emph{Query count.}
Each controlled execution of $S_j$ queries $O_V\oplus O_I$ once.
This direct-sum oracle uses at most one controlled call each to
$O_V$ and $O_I$.  A call to $O_V=V\oplus V^\dagger$ can be implemented using
$\cO(1/\delta)$ calls to
$U_H,U_H^\dagger$ by Lemmas~\ref{lem:spectral-test}
and~\ref{lem:hamiltonian-rescaling}; a call to
$O_I=U_I\oplus U_I^\dagger$ uses at most one call each to
$U_I,U_I^\dagger$.  Over $K=\cO(j)$ executions, $C_j$ therefore uses $\cO(j/\delta)$ calls to
$U_H,U_H^\dagger$ and $\cO(j)$ calls to $U_I,U_I^\dagger$.

\emph{Gate count.}
After truncation, the spaces in~\eqref{eq:amplification-counter-spaces}
are encoded by index registers and one
$\log_2D$-qubit counter register.
In this gate-count analysis, $S_g$ and $S_j$ denote the transducers obtained
by replacing $S_{\rm pur}$ with $S_{\rm pur}^{[D]}$ in their constructions.
The work-unitary gate costs obey the
following bounds:
\begin{center}
\begin{tabular}{@{}l@{\qquad}l@{}}
\toprule
Gate-cost relation & Bound\\
\midrule
$T(S_\phi)$ & $\cO(n+1)$\\
$T(S_{\rm ref})$ & $\cO(a_H+1)$\\
$T(S_V)+T(S_{V^\dagger})$ & $\cO(1)$\\
$T(S_{\rm pur}^{[D]})$ & $\cO(\log D)$\\
\midrule
$T(S_{\rm pur}^{[D]}\circ S_{\rm ref})
 =\cO\bigl(T(S_{\rm pur}^{[D]})+T(S_{\rm ref})+\log D\bigr)$
 & $\cO(a_H+\log D)$\\
\midrule
$T(S_g)
 =\cO\bigl(T(S_V)+T(S_{V^\dagger})
    +T(S_{\rm pur}^{[D]}\circ S_{\rm ref})+1\bigr)$
 & $\cO(a_H+\log D)$\\
\midrule
$T(S_j)=\cO\bigl(T(S_g)+T(S_\phi)+\log(j+1)\bigr)$
 & $\cO(n+a_H+\log M)$\\
\bottomrule
\end{tabular}
\end{center}
The circuits for $R_\phi=U_I(2\proj0-I)U_I^\dagger$ and
$O_{\rm ref}$ in~\eqref{eq:spectral-test-subspaces} use zero-state
reflections on $n$ and $a_H+\cO(1)$ qubits, respectively.  Their gate
counts are $\cO(n)$ and $\cO(a_H+1)$: a zero-state reflection uses a
multi-controlled NOT conjugated by single-qubit gates.
Lemma~\ref{lem:circuit-to-transducer}
then gives the first two rows of the table.  The $S_V,S_{V^\dagger}$ and
$S_{\rm pur}^{[D]}$ rows follow from
Lemmas~\ref{lem:circuit-to-transducer} and~\ref{lem:finite-counter},
respectively.

By Proposition~\ref{prop:functional-composition}, the work unitary of
$S_{\rm pur}^{[D]}\circ S_{\rm ref}$ applies
$I_{\mathcal E}\otimes S_{\rm ref}^\circ$ and $(S_{\rm pur}^{[D]})^\circ$.
By~\eqref{eq:purifier-oracle-label-space}, the control on
$S_{\rm ref}^\circ$ tests whether $\mathsf T\in\{0,1\}$ and
$\mathsf J\ne0$.  The clock test has constant cost; computing and
uncomputing the nonzero test on $\mathsf J$ adds $\cO(\log D)$ gates.
Controlling the input-oracle call uses the same test and adds another
$\cO(\log D)$ gates, which are not included in
$T(S_{\rm pur}^{[D]}\circ S_{\rm ref})$.
The three-factor composition in~\eqref{eq:ground-reflection-transducer}
has a constant-size index, giving the $S_g$ row.

For $S_j$, the $2j$ summands in
\eqref{eq:fixed-branch-work-unitary} alternate $S_g^\circ$ and
$S_\phi^\circ$.  Numbering them $0,\ldots,2j-1$, even labels select
$S_g^\circ$ and odd labels select $S_\phi^\circ$.  Thus the least
significant bit of the $\cO(\log(j+1))$-qubit index selects the work
unitary, and the remaining bits encode $t-1$.  The controlled increment
modulo $2j$ in $P$ uses $\cO(\log(j+1))$ gates.  Thus the
selection among the $2j$ terms adds only $\cO(\log(j+1))$ gates, rather
than a factor of $j$.  With $j<M$ and $D=\cO(M)$, this proves
$T(S_j)=\cO(n+a_H+\log M)$.  The $\cO(\log D)$ gates controlling
the oracle call obey the same bound.

Each of the $K=\cO(j)$ executions of $S_j$ contains at most one call
to $O_V=V\oplus V^\dagger$.  Implementing this call uses
$\cO((a_H+1)/\delta)$ one- and two-qubit gates by
Lemma~\ref{lem:spectral-test}.  Lemma~\ref{lem:circuit-implementation}
adds $\cO(K)$ gates outside these executions.  Hence the total is
\[
\cO\!\left(K(n+a_H+\log M)+K\frac{a_H+1}{\delta}\right)
=\cO\!\left(j(n+a_H+\log M)+\frac{(a_H+1)j}{\delta}\right),
\]
as in~\eqref{eq:fixed-branch-gate-cost}.
\end{proof}

\noindent\begin{minipage}{\linewidth}
We use the circuits $C_j$ in the randomized procedure of
Section~\ref{subsec:ideal-randomized-amplification}, as described in
Algorithm~\ref{alg:constant-precision-amplification}.
\begin{algorithm}[H]
\caption{Ground-state overlap amplification}
\label{alg:constant-precision-amplification}
\begin{algorithmic}[1]
\Require $\gamma$, $\delta$, $U_I$, and a $1$-block-encoding of $H$
\Ensure A state on the system register $\mathsf S$
\State Initialize $\mathsf S$ in $\ket0_{\mathsf S}$ and apply $U_I$ to prepare $\ket\phi$
\State $M\gets\lceil2/\gamma\rceil$
\State Sample $c$ uniformly from $\{0,1\}$
\If{$c=1$}
  \State Sample $j$ uniformly from $\{0,\ldots,M-1\}$
  \State Initialize $\mathsf A\mathsf R$ in $\ket0_{\mathsf A\mathsf R}$
  \State Apply $C_j$ to $\mathsf A\mathsf R\mathsf S$ with its work registers initialized to zero
    \Comment{Proposition~\ref{prop:fixed-branch-circuit}}
  \State Discard all registers except $\mathsf S$
\EndIf
\State \Return $\mathsf S$
\end{algorithmic}
\end{algorithm}
\end{minipage}
\par\smallskip

We now check that averaging the fixed-$j$ circuits retains the overlap
guarantee of Lemma~\ref{lem:ideal-averaged-grover-state}.
\begin{proposition}[Overlap amplification]
\label{prop:constant-ground-weight}
Let $\rho_c$ be the system state returned by
Algorithm~\ref{alg:constant-precision-amplification}, averaged over its
random choices.  Then
\[
\bra{\psi_0}\rho_c\ket{\psi_0}\ge\frac{11}{64}>\frac18.
\]
Its worst-case query counts satisfy $Q_I=\cO(1/\gamma)$ and
$Q_H=\cO(1/(\gamma\delta))=\cO(\alpha/(\gamma\Delta))$.
If $\mathsf S$ has $n$ qubits, the number of additional one- and
two-qubit gates in the algorithm is
\begin{equation}
\cO\!\left(
\frac{a_H+1}{\gamma\delta}
+\frac{n+a_H+\log(2/\gamma)}{\gamma}
\right).
\label{eq:coarse-amplification-gate-cost}
\end{equation}
\end{proposition}

\begin{proof}
We compare the actual output of the amplification algorithm with the ideal
average state in~\eqref{eq:ideal-amplified-mixture}.
Averaging over the classical choices in the algorithm gives
\[
\rho_c=\frac12\proj\phi+\frac1{2M}\sum_{j=0}^{M-1}\rho_j.
\]
By convexity, Proposition~\ref{prop:fixed-branch-circuit}, and
\eqref{eq:ideal-amplified-mixture},
\[
\Dtr(\rho_c,\rho_{\mathrm{ideal}})
\le\frac1{2M}\sum_{j=0}^{M-1}
\Dtr\!\left(\rho_j,G^j\proj\phi(G^\dagger)^j\right)
\le\frac1{64}.
\]
Applying the variational bound for trace distance to $\proj{\psi_0}$
and then using Lemma~\ref{lem:ideal-averaged-grover-state}, we obtain
\[
\begin{aligned}
\bra{\psi_0}\rho_c\ket{\psi_0}
&\ge \bra{\psi_0}\rho_{\mathrm{ideal}}\ket{\psi_0}
   -\Dtr(\rho_c,\rho_{\mathrm{ideal}})\\
&\ge\frac3{16}-\frac1{64}=\frac{11}{64}.
\end{aligned}
\]
The preparation of $\ket\phi$ uses one call to $U_I$.  When $c=1$,
Proposition~\ref{prop:fixed-branch-circuit} bounds the
calls made by $C_j$ by $\cO(j/\delta)$ to the Hamiltonian oracle and
$\cO(j)$ to the state-preparation oracle.  Since $j<M\le3/\gamma$, these
bounds give the stated $Q_H$ and $Q_I$.  Substituting $j<M$ and
$\log M=\cO(\log(2/\gamma))$ into~\eqref{eq:fixed-branch-gate-cost}, and
using $\delta^{-1}=\cO(\alpha/\Delta)$, gives
\eqref{eq:coarse-amplification-gate-cost}.
\end{proof}

\section{Final filtering}
\label{sec:final-filtering}

Proposition~\ref{prop:constant-ground-weight} gives a state $\rho_c$ with
$\bra{\psi_0}\rho_c\ket{\psi_0}\ge1/8$.  We first use one spectral filter
to obtain an $\varepsilon$-accurate state with constant success probability.
We then reduce the failure probability to $\zeta$ while keeping a bound
on the queries made in every run.

\subsection{Constant-success preparation}
\label{subsec:constant-success-preparation}

\noindent\begin{minipage}{\linewidth}
Algorithm~\ref{alg:ground-state-preparation-run} applies one spectral filter
with accuracy $\varepsilon/8$ to $\rho_c$ and accepts its all-zero ancilla
outcome.

\begin{algorithm}[H]
\caption{One run of ground-state preparation}
\label{alg:ground-state-preparation-run}
\begin{algorithmic}[1]
\Require $\varepsilon$, $\gamma$, $\delta$, $U_I$, and a $1$-block-encoding of $H$
\Ensure Either \textsc{success} with a system register, or \textsc{failure}
\State Run Algorithm~\ref{alg:constant-precision-amplification} to prepare
  $\rho_c$ in the system register
\State Construct $U_{\varepsilon/8,\delta}$
  \Comment{Lemma~\ref{lem:spectral-filter}}
\State Initialize $\mathsf F$ in $\ket0_{\mathsf F}$ and apply
  $U_{\varepsilon/8,\delta}$ jointly to $\mathsf F$ and the system register
\State Measure $\mathsf F$
\If{the outcome is all zero}
  \State \Return \textsc{success} and the system register
\Else
  \State \Return \textsc{failure}
\EndIf
\end{algorithmic}
\end{algorithm}
\end{minipage}
\par\smallskip

We first bound the final filter's success probability and the error of the
accepted state.
\begin{proposition}[Correctness of the final filter]
\label{prop:final-filter-correctness}
Let $\rho_c$ be any density operator satisfying
\[
\bra{\psi_0}\rho_c\ket{\psi_0}\ge\frac18.
\]
For $0<\varepsilon\le1/10$, the filtering part of
Algorithm~\ref{alg:ground-state-preparation-run} satisfies
\begin{equation}
\Pr[\mathrm{success}]\ge\frac9{128}>\frac1{16},
\qquad
\Dtr(\rho_{\rm succ},\proj{\psi_0})\le\varepsilon.
\label{eq:final-filter-guarantee}
\end{equation}
\end{proposition}

\begin{proof}
By Lemma~\ref{lem:spectral-filter},
\[
(\bra0_{\mathsf F}\otimes I)U_{\varepsilon/8,\delta}
(\ket0_{\mathsf F}\otimes I)
=f(H)=\sum_{k\ge0}f(E_k)\proj{\psi_k},
\]
where $|f(E_0)|\ge1-\varepsilon/8$ and
$|f(E_k)|\le\varepsilon/8$ for $k\ge1$.

The probability of the $\ket0_{\mathsf F}$ outcome is
\begin{align*}
s:=\Pr[\mathrm{success}]
&=\Tr\!\bigl(f(H)\rho_c f(H)^\dagger\bigr)
=\Tr\!\bigl(\rho_c f(H)^\dagger f(H)\bigr)\\
&=\sum_{k\ge0}|f(E_k)|^2
  \bra{\psi_k}\rho_c\ket{\psi_k}
\ge(1-\varepsilon/8)^2\bra{\psi_0}\rho_c\ket{\psi_0}\\
&\ge\frac{(1-\varepsilon/8)^2}{8}
\ge\frac9{128}>\frac1{16}.
\end{align*}
Conditioned on this outcome, the system state is
$\rho_{\rm succ}=f(H)\rho_c f(H)^\dagger/s$.  Hence
\begin{align*}
1-\bra{\psi_0}\rho_{\rm succ}\ket{\psi_0}
&=\frac1s\sum_{k\ge1}|f(E_k)|^2
  \bra{\psi_k}\rho_c\ket{\psi_k}\\
&\le\frac{(\varepsilon/8)^2}{s}
\le\frac{\varepsilon^2}{8(1-\varepsilon/8)^2}.
\end{align*}
The Fuchs--van de Graaf inequality~\cite{FvdG99} now gives
\[
\Dtr(\rho_{\rm succ},\proj{\psi_0})
\le\sqrt{1-\bra{\psi_0}\rho_{\rm succ}\ket{\psi_0}}
\le\frac{\varepsilon}{2\sqrt2(1-\varepsilon/8)}
\le\frac{\sqrt2}{3}\varepsilon<\varepsilon.
\]
\end{proof}

We next bound the query and gate costs of the final filter in
Algorithm~\ref{alg:ground-state-preparation-run}.
\begin{proposition}[Cost of the final filter]
\label{prop:final-filter-cost}
The filtering part of Algorithm~\ref{alg:ground-state-preparation-run} makes
no state-preparation query and uses
\begin{equation}
\cO\!\left(\frac1\delta\log\frac1\varepsilon\right)
=\cO\!\left(\frac\alpha\Delta\log\frac1\varepsilon\right)
\label{eq:final-filter-query-cost}
\end{equation}
Hamiltonian queries.  If the block-encoding of the normalized Hamiltonian has
$a_H$ ancilla qubits, it also uses
$\cO((a_H+1)\delta^{-1}\log(1/\varepsilon))$ additional one- and two-qubit gates.
\end{proposition}

\begin{proof}
Lemma~\ref{lem:spectral-filter} gives
$\cO(\delta^{-1}\log(1/\varepsilon))$ queries to the normalized block-encoding
and the stated gate bound.
Each query to the normalized block-encoding or its inverse uses one
controlled query to $U_H$ or $U_H^\dagger$, respectively, by
Lemma~\ref{lem:hamiltonian-rescaling}, and
$\delta^{-1}=\cO(\alpha/\Delta)$.
\end{proof}

The filtering and amplification bounds together give the following guarantees
for Algorithm~\ref{alg:ground-state-preparation-run}.
\begin{proposition}[Constant-success preparation]
\label{prop:constant-success-preparation}
Algorithm~\ref{alg:ground-state-preparation-run} succeeds with probability at
least $1/16$ and, conditioned on success, has trace-distance error at most
$\varepsilon$.  Its worst-case query counts satisfy
\[
Q_I=\cO\!\left(\frac1\gamma\right),
\qquad
Q_H=\cO\!\left(
\frac1{\gamma\delta}
+\frac1\delta\log\frac1\varepsilon
\right).
\]
\end{proposition}

\begin{proof}
By Proposition~\ref{prop:constant-ground-weight}, the state $\rho_c$
prepared in Step~1 of Algorithm~\ref{alg:ground-state-preparation-run} satisfies the
hypothesis of Proposition~\ref{prop:final-filter-correctness}.  That
proposition gives success probability at least $1/16$ and conditional
trace-distance error at most $\varepsilon$.

Combining Propositions~\ref{prop:constant-ground-weight}
and~\ref{prop:final-filter-cost}
therefore gives
\[
Q_I=\cO\!\left(\frac1\gamma\right),
\qquad
Q_H=\cO\!\left(
\frac1{\gamma\delta}
+\frac1\delta\log\frac1\varepsilon
\right).
\]
Finally, Lemma~\ref{lem:hamiltonian-rescaling} gives
$\delta^{-1}=\cO(\alpha/\Delta)$; substituting this relation expresses the
bounds in the original parameters.
\end{proof}

Proposition~\ref{prop:constant-success-preparation} also gives the
expected-query upper bounds in Theorem~\ref{thm:expected-query-complexity}
by repetition.

\begin{proof}[Proof of the upper bounds in Theorem~\ref{thm:expected-query-complexity}]
Run Algorithm~\ref{alg:ground-state-preparation-run} independently until
success.  Each run succeeds with probability at least $1/16$, so the mean
number of runs is at most $16$.  The per-run query bounds in
Proposition~\ref{prop:constant-success-preparation} therefore give the
expected query counts.  The output state is the state of one run conditioned
on success, so it has the same trace-distance guarantee.
\end{proof}

\subsection{Reducing the failure probability}
\label{subsec:bounded-query-preparation}

To reduce the failure probability without repeating the highest-accuracy
filter after every rejection, we apply filters of increasing accuracy and
restart as soon as one rejects.  Let
\[
K:=\min\{j\ge0:16^{-2^j}\le\varepsilon/4\},
\qquad
\eta_j:=16^{-2^j},\quad 0\le j\le K.
\]
By Lemma~\ref{lem:spectral-filter} and the query accounting in
Lemma~\ref{lem:hamiltonian-rescaling}, there is an absolute constant $C$
such that $U_{\eta_j,\delta}$ uses at most $(C/\delta)2^j$ Hamiltonian
queries.

\noindent\begin{minipage}{\linewidth}
In Algorithm~\ref{alg:increasing-accuracy-filtering}, the remaining filtering
budget $B$ is measured in units of $C/\delta$ queries, so applying
$U_{\eta_j,\delta}$ decreases $B$ by $2^j$.

\begin{algorithm}[H]
\caption{Ground-state preparation with reduced failure probability}
\label{alg:increasing-accuracy-filtering}
\begin{algorithmic}[1]
\Require $0<\varepsilon\le\zeta\le1/10$, $\gamma$, $\delta$, $U_I$, and a $1$-block-encoding of $H$
\Ensure Either \textsc{success} with the system register $\mathsf S$, or \textsc{failure}
\State $R\gets\left\lceil\log(2/\zeta)/\log(16/15)\right\rceil$
\State $B\gets 2^{K+1}-1+2R+\left\lceil\log_2(2/\zeta)\right\rceil$
  \Comment{Remaining filtering budget}
\For{$t=1,\ldots,R$}
  \State Run Algorithm~\ref{alg:constant-precision-amplification} independently
    to prepare $\rho_c$ in $\mathsf S$
  \For{$j=0,\ldots,K$}
    \If{$B<2^j$}
      \State \Return \textsc{failure}
    \EndIf
    \State $B\gets B-2^j$
    \State Initialize a fresh $\mathsf F$ in $\ket0_{\mathsf F}$ and apply
      $U_{\eta_j,\delta}$ to $\mathsf F\mathsf S$
    \State Measure $\mathsf F$ in the computational basis
    \If{the outcome is not all zero}
      \State Discard $\mathsf S$ and continue with the next $t$
    \EndIf
  \EndFor
  \State \Return \textsc{success} and $\mathsf S$
\EndFor
\State \Return \textsc{failure}
\end{algorithmic}
\end{algorithm}
\end{minipage}
\par\smallskip

We now prove the success and error guarantees and bound both the queries
and the one- and two-qubit gates used by this algorithm.

\begin{proposition}[Preparation with bounded queries]
\label{prop:increasing-accuracy-filtering}
Let $0<\varepsilon\le\zeta\le1/10$.  Algorithm~\ref{alg:increasing-accuracy-filtering}
succeeds with probability at least $1-\zeta$ and, conditioned on success,
satisfies
$\Dtr(\rho_{\rm succ},\proj{\psi_0})\le\varepsilon$.  Its worst-case query
counts satisfy
\[
Q_I=\cO\!\left(\frac1\gamma\log\frac1\zeta\right),
\qquad
Q_H=\cO\!\left(
\frac1{\gamma\delta}\log\frac1\zeta
+\frac1\delta\log\frac1\varepsilon
\right).
\]
If $\mathsf S$ has $n$ qubits, the number of additional one- and two-qubit
gates in the worst case is
\begin{equation}
\cO\!\left[
\frac{a_H+1}{\delta}
\left(\frac1\gamma\log\frac1\zeta+\log\frac1\varepsilon\right)
+\frac{n+a_H+\log(2/\gamma)}{\gamma}\log\frac1\zeta
\right].
\label{eq:increasing-accuracy-gate-cost}
\end{equation}
\end{proposition}

\begin{proof}
\emph{Acceptance probability and output error.}
We first compute the probability of passing all filters and the resulting
state.  By Lemma~\ref{lem:spectral-filter},
\[
f_j(H):=(\bra0_{\mathsf F}\otimes I)U_{\eta_j,\delta}
(\ket0_{\mathsf F}\otimes I)
=\sum_{k\ge0}f_j(E_k)\proj{\psi_k},
\]
where
\begin{equation}
|f_j(E_0)-1|\le\eta_j,
\qquad
|f_j(E_k)|\le\eta_j\quad\text{for }k\ge1.
\label{eq:multistage-filter-bounds}
\end{equation}
Set $F_j(H):=f_j(H)\cdots f_0(H)$.
When all $K+1$ measurements of $\mathsf F$ give outcome $0$, the
unnormalized output on $\mathsf S$ is
\[
f_K(H)\cdots f_0(H)\rho_c f_0(H)^\dagger\cdots f_K(H)^\dagger
=F_K(H)\rho_cF_K(H)^\dagger.
\]
Write $p:=\bra{\psi_0}\rho_c\ket{\psi_0}\ge1/8$.  The probability that
all measurements give $0$ is
\begin{align}
p_{\rm pass}
&=\Tr\!\left(F_K(H)\rho_cF_K(H)^\dagger\right)\notag\\
&=\sum_{k\ge0}\bra{\psi_k}\rho_c\ket{\psi_k}
  \prod_{j=0}^K|f_j(E_k)|^2\notag\\
&\ge p\prod_{j=0}^K|f_j(E_0)|^2
\ge p\prod_{j=0}^K(1-\eta_j)^2\notag\\
&\ge p\left(1-\sum_{j=0}^K\eta_j\right)^2
\ge\frac18\left(\frac{14}{15}\right)^2
>\frac1{16}.
\label{eq:multistage-pass-probability}
\end{align}
The third line uses~\eqref{eq:multistage-filter-bounds}; the last line uses
$\sum_{j=0}^K\eta_j<\sum_{r\ge1}16^{-r}=1/15$.
The conditional state is
\[
\rho_{\rm pass}
:=\frac{F_K(H)\rho_cF_K(H)^\dagger}{p_{\rm pass}}.
\]
Using~\eqref{eq:multistage-filter-bounds}, we bound its excited-state
probability as
\begin{align*}
1-\bra{\psi_0}\rho_{\rm pass}\ket{\psi_0}
&=\sum_{k\ge1}\bra{\psi_k}\rho_{\rm pass}\ket{\psi_k}\\
&=\frac1{p_{\rm pass}}\sum_{k\ge1}
  \bra{\psi_k}\rho_c\ket{\psi_k}
  \prod_{j=0}^K|f_j(E_k)|^2\\
&\le\frac{\prod_{j=0}^K\eta_j^2}{p_{\rm pass}}
  \sum_{k\ge1}\bra{\psi_k}\rho_c\ket{\psi_k}\\
&=\frac{(1-p)\prod_{j=0}^K\eta_j^2}{p_{\rm pass}}
\le16\prod_{j=0}^K\eta_j^2.
\end{align*}
Here $\Tr\rho_c=1$, and the final bound uses $p_{\rm pass}>1/16$
from~\eqref{eq:multistage-pass-probability}.
The Fuchs--van de Graaf inequality therefore gives
\[
\Dtr(\rho_{\rm pass},\proj{\psi_0})
\le4\prod_{j=0}^K\eta_j
\le4\eta_K
\le\varepsilon.
\]

\emph{Oracle calls made by filters in rejected attempts.}
Let $q_j$ be the
probability that the measurements for filters $0,\ldots,j-1$ all give $0$
and the measurement for filter $j$ does not.  For $j\ge1$, the preceding
zero outcomes have unnormalized system output
$F_{j-1}(H)\rho_cF_{j-1}(H)^\dagger$.  Subtracting the probability that
filter $j$ also gives $0$ gives
\begin{align}
q_j
&=\Tr\!\left(F_{j-1}(H)\rho_cF_{j-1}(H)^\dagger\right)
  -\Tr\!\left(F_j(H)\rho_cF_j(H)^\dagger\right)\notag\\
&=\Tr\!\left[(I-f_j(H)^\dagger f_j(H))
  F_{j-1}(H)\rho_cF_{j-1}(H)^\dagger\right]\notag\\
&=\sum_{k\ge0}\bra{\psi_k}\rho_c\ket{\psi_k}
  \prod_{\ell<j}|f_\ell(E_k)|^2
  \bigl(1-|f_j(E_k)|^2\bigr)\notag\\
&=p\prod_{\ell<j}|f_\ell(E_0)|^2
  \bigl(1-|f_j(E_0)|^2\bigr)
  +\sum_{k\ge1}\bra{\psi_k}\rho_c\ket{\psi_k}
    \prod_{\ell<j}|f_\ell(E_k)|^2
    \bigl(1-|f_j(E_k)|^2\bigr)\notag\\
&\le2p\eta_j+\left(\prod_{\ell<j}\eta_\ell^2\right)
  \sum_{k\ge1}\bra{\psi_k}\rho_c\ket{\psi_k}\notag\\
&=2p\eta_j+(1-p)\prod_{\ell<j}\eta_\ell^2\notag\\
&\le2\eta_j+16^{-2(2^j-1)}
\le3\eta_j=3\cdot16^{-2^j}.
\label{eq:late-rejection-probability}
\end{align}
The first inequality uses $|f_\ell(E_k)|\le1$ from
Lemma~\ref{lem:spectral-filter}, \eqref{eq:multistage-filter-bounds}, and
$1-|f_j(E_0)|^2\le1-(1-\eta_j)^2\le2\eta_j$.
The last inequality uses
$2(2^j-1)\ge2^j$ for $j\ge1$.

To bound the total number of calls to $U_H$ and $U_H^\dagger$ made by the
filters in rejected attempts, let $X$ be the amount subtracted from $B$
during an attempt if it is rejected, and set
$X=0$ if it is accepted.  Rejection at filter $j$ gives
\[
X=\sum_{\ell=0}^j2^\ell=2^{j+1}-1.
\]
The amount subtracted for the accepted attempt will be added separately.
The small probabilities of late rejection keep $\mathbb E[2^X]$ bounded:
using $q_0\le1$ and~\eqref{eq:late-rejection-probability},
\begin{align}
\mathbb E[2^X]
&=p_{\rm pass}+2q_0
  +\sum_{j=1}^K q_j2^{2^{j+1}-1}\notag\\
&\le1+2+\sum_{j\ge1}
3\cdot16^{-2^j}2^{2^{j+1}-1}\notag\\
&=3+\frac32\sum_{j\ge1}4^{-2^j}
\le3+\frac32\sum_{r\ge2}4^{-r}<4.
\label{eq:failed-filter-moment}
\end{align}

\emph{Failure probability and query counts.}
Consider $R$ independent attempts without the budget check in
Algorithm~\ref{alg:increasing-accuracy-filtering}.
By~\eqref{eq:multistage-pass-probability}, the probability that none is
accepted is at most $(15/16)^R\le\zeta/2$.
Writing $X_t$ for the value of $X$ in attempt $t$, independence gives
\[
\mathbb E\!\left[2^{\sum_{t=1}^R X_t}\right]
=\prod_{t=1}^R\mathbb E[2^{X_t}]
<4^R.
\]
Markov's inequality therefore gives
\begin{equation}
\Pr\!\left[
\sum_{t=1}^R X_t>2R+\log_2\frac2\zeta
\right]
\le\frac{\mathbb E\!\left[2^{\sum_{t=1}^R X_t}\right]}
{2^{2R+\log_2(2/\zeta)}}
<\frac{4^R}{4^R(2/\zeta)}=\frac\zeta2.
\label{eq:failed-filter-tail}
\end{equation}

An accepted attempt subtracts $2^{K+1}-1$ from $B$.
If at least one of the $R$ attempts is accepted and
$\sum_{t=1}^R X_t\le2R+\log_2(2/\zeta)$, the total amount subtracted
through the first accepted attempt is at most
\[
2^{K+1}-1+\sum_{t=1}^R X_t
\le2^{K+1}-1+2R+\log_2\frac2\zeta,
\]
which does not exceed the initial value of $B$.  All deductions are
nonnegative, so the budget check cannot halt the algorithm before the first
accepted attempt finishes.  A union bound and~\eqref{eq:failed-filter-tail} now give
\[
\Pr[\mathrm{failure}]
\le(1-p_{\rm pass})^R
  +\Pr\!\left[\sum_{t=1}^R X_t>2R+\log_2\frac2\zeta\right]
\le\frac\zeta2+\frac\zeta2=\zeta.
\]

An attempt can return success only if the budget available at its start
is sufficient to apply all $K+1$ filters.
The budget depends only on preceding attempts, while $\rho_c$ is prepared
independently, so conditioning on sufficient budget leaves the input state
$\rho_c$ unchanged.  Hence $\rho_{\rm succ}=\rho_{\rm pass}$ and its
trace-distance error is at most $\varepsilon$.

Finally, the algorithm prepares $\rho_c$ at most
$R=\cO(\log(1/\zeta))$ times.  By
Proposition~\ref{prop:constant-ground-weight}, these preparations use
$\cO(R/\gamma)$ state-preparation queries and
$\cO(R/(\gamma\delta))$ Hamiltonian queries in total.
The total number of Hamiltonian queries used by the filters is at most
$C/\delta$ times the initial value of $B$.
Since $\eta_{K-1}>\varepsilon/4$, minimality of
$K$ gives $2^K<2\log_{16}(4/\varepsilon)=\cO(\log(1/\varepsilon))$.
Consequently,
\begin{align*}
Q_I&=\cO(R/\gamma)
=\cO\!\left(\frac1\gamma\log\frac1\zeta\right),\\
Q_H&=\cO\!\left(\frac{R}{\gamma\delta}
   +\frac{2^K+R+\log(1/\zeta)}\delta\right)\\
&=\cO\!\left(\frac1{\gamma\delta}\log\frac1\zeta
   +\frac1\delta\log\frac1\varepsilon\right),
\end{align*}
where the last line uses $\gamma\le1$.

\emph{Gate count.}
There are at most $R$ preparations of $\rho_c$, each with the gate count in
\eqref{eq:coarse-amplification-gate-cost}.
By Lemma~\ref{lem:spectral-filter}, applying $U_{\eta_j,\delta}$ uses
$\cO((a_H+1)2^j/\delta)$ additional one- and two-qubit gates.
The sum of $2^j$ over all applied filters is at most the initial value of
$B$.  Hence the total gate count is
\[
\cO\!\left[
R\left(\frac{a_H+1}{\gamma\delta}
  +\frac{n+a_H+\log(2/\gamma)}{\gamma}\right)
+\frac{a_H+1}{\delta}\left(2^K+R+\log\frac1\zeta\right)
\right].
\]
Substituting $R=\cO(\log(1/\zeta))$ and
$2^K=\cO(\log(1/\varepsilon))$, and using $\gamma\le1$, gives
\eqref{eq:increasing-accuracy-gate-cost}.
\end{proof}

We now express the query bounds of
Proposition~\ref{prop:increasing-accuracy-filtering} in the original
Hamiltonian parameters to prove Theorem~\ref{thm:upper-bound}.

\begin{proof}[Proof of Theorem~\ref{thm:upper-bound}]
Apply Proposition~\ref{prop:increasing-accuracy-filtering} and use
$\delta^{-1}=\cO(\alpha/\Delta)$ from
Lemma~\ref{lem:hamiltonian-rescaling}.
\end{proof}

We next use Theorem~\ref{thm:upper-bound} to bound the trace-distance error
of the average output state, including the state returned on failure.

\begin{proof}[Proof of the upper bounds in Theorem~\ref{thm:unconditional-output}]
Apply Theorem~\ref{thm:upper-bound} with error target $\varepsilon/2$ on
success and failure probability $\zeta=\varepsilon/2$.  On
failure, return an arbitrary system state and omit the flag.  The resulting
average output state $\rho$ is a mixture of $\rho_{\rm succ}$ and the state returned on
failure.  Convexity of trace distance gives
\[
\Dtr(\rho,\proj{\psi_0})
\le\frac\varepsilon2+\Pr[\mathrm{failure}]
\le\varepsilon.
\]
Since $\gamma\le1$, the query bounds in
Theorem~\ref{thm:upper-bound} reduce to the bounds stated in
Theorem~\ref{thm:unconditional-output}.
\end{proof}

\section{Optimality of the query bounds}
\label{sec:optimality}

The first two lemmas give lower bounds on $Q_H$ for algorithms with constant
success probability and are used
to prove Theorem~\ref{thm:expected-query-complexity}.  Both bounds hold with
unrestricted access to $U_I$ and $U_I^\dagger$.  We next obtain the
worst-case lower bound on $Q_H$ from Somma and de Wolf~\cite{SdW26}, then
prove the lower bounds on $U_I$ queries in both settings.

We work first with exact $1$-block-encodings of Hamiltonians $\widetilde H$
satisfying
\begin{equation}
\norm{\widetilde H}\le1,
\qquad
\widetilde E_0\le-\delta,
\qquad
\widetilde E_k\ge\delta\quad\text{for }k\ge1,
\qquad
|\langle\psi_0|U_I|0\rangle|\ge\gamma.
\label{eq:lower-normalized-promise}
\end{equation}
To recover the parameters of Section~\ref{sec:model}, set
\begin{equation}
H:=\alpha\widetilde H,
\qquad
\mu:=0,
\qquad
\Delta:=2\alpha\delta.
\label{eq:lower-bound-rescaling}
\end{equation}
The same oracle $U_H$ is then an $\alpha$-block-encoding of $H$, and
\eqref{eq:lower-normalized-promise} becomes the spectral promise
\eqref{eq:separator-promise}.  We therefore
prove the technical bounds in terms of $\delta$ and use
$\delta=\Delta/(2\alpha)$ when returning to the statements in
Section~\ref{sec:model}.

For the hybrid arguments, we write the algorithm in the standard form for
quantum query algorithms~\cite{BBBV97}.  Let $\mathsf B$ be the
success-flag qubit, and $\mathsf W$ the remaining workspace.  Classical
random choices and intermediate measurement outcomes can be kept coherently
in $\mathsf W$.  Branches with fewer queries use controls in $\mathsf W$ to
disable the remaining calls.  If $x$ labels the input oracles, then after $m$
oracle calls the state before the final measurement can be written as
\begin{equation}
\ket{\Psi_m^x}_{\mathsf B\mathsf S\mathsf W}
=A_m\widehat U_{m,x}A_{m-1}\cdots
A_1\widehat U_{1,x}A_0\ket0_{\mathsf B\mathsf S\mathsf W},
\label{eq:lower-query-circuit}
\end{equation}
where $A_0,\ldots,A_m$ are oracle-independent unitaries and
\[
\widehat U_{q,x}
=\proj0_{\mathsf Q}\otimes I
+\proj1_{\mathsf Q}\otimes U_{q,x},
\qquad
U_{q,x}\in
\{U_{H,x},U_{H,x}^\dagger,U_{I,x},U_{I,x}^\dagger\}.
\]
Here $\mathsf Q$ is a one-qubit control in $\mathsf W$; identity operators on
registers not acted on by $U_{q,x}$ are implicit.  Setting $\mathsf Q$ to
$\ket1$ gives an uncontrolled call.  Measuring $\mathsf B$ gives the flag
$b$ from Section~\ref{sec:model}; conditioned on $b=1$, tracing out
$\mathsf B\mathsf W$ gives $\rho_{\rm succ}$ on $\mathsf S$.

\subsection{Hamiltonian queries}
\label{subsec:hamiltonian-query-lower}

\subsubsection{Dependence on the initial overlap}
\label{subsec:overlap-lower}

Arihara and Murao~\cite[Corollary~7]{AM26} establish the joint overlap--gap
lower bound for constant-error ground-state preparation with block-encoding
access.  We use the following version for the access and output model of
Section~\ref{sec:model}.

\begin{lemma}[Overlap lower bound]
\label{lem:overlap-lower}
Fix \(0<s_0\le1\) and \(0<\varepsilon\le9/20\).  Let
\(0<\gamma\le1/\sqrt2\) and \(0<\delta<1\).  Suppose an algorithm succeeds
with probability at least $s_0$ and satisfies
$\Dtr(\rho_{\rm succ},\proj{\psi_0})\le\varepsilon$ on every instance
satisfying~\eqref{eq:lower-normalized-promise}.  If it uses at most $T$
Hamiltonian queries, then
\[
T=\Omega_{s_0}\!\left(\frac1{\gamma\delta}\right)
\]
even with unrestricted access to $U_I$ and $U_I^\dagger$.
\end{lemma}

A proof covering controlled Hamiltonian queries is given in
Appendix~\ref{app:overlap-lower}.

\subsubsection{Dependence on the target precision}
\label{subsec:precision-lower}

We adapt the trigonometric-polynomial argument of Mande and de
Wolf~\cite[Claims~5.2--5.3]{MdW26} to block-encoding queries and
success-conditioned ground-state preparation.

\begin{lemma}[Precision lower bound]
\label{lem:precision-lower}
Let $0<s_0\le1$, $0<\varepsilon<1/2$, $0<\gamma\le1/\sqrt2$, and
$0<\delta\le1/16$.  Suppose an algorithm succeeds with probability at
least $s_0$ and satisfies
$\Dtr(\rho_{\rm succ},\proj{\psi_0})\le\varepsilon$ on every instance
satisfying~\eqref{eq:lower-normalized-promise}.  If it uses at most $T$
Hamiltonian queries, then
\[
T\ge\frac1{448\delta}\log\frac{s_0^2}{4\varepsilon}
\]
even with unrestricted access to $U_I$ and $U_I^\dagger$.
\end{lemma}

\begin{proof}
We compare two Hamiltonians whose ground states differ near $\theta=0$ but
agree when $|\sin\theta|$ is large.  For $\theta\in[-\pi,\pi)$ and $h>0$,
define
\[
H_{\theta,\pm}:=\frac{\sin\theta\pm2h}{1+2h}\,Z_{\mathsf S}.
\]
Write $c:=1+2h$.  Both Hamiltonians have norm at most one, since
$|\sin\theta\pm2h|\le1+2h=c$.

These Hamiltonians have exact $1$-block-encodings
\[
B_{\theta,\pm}:=(U_h^\dagger\otimes I_{\mathsf A\mathsf S})
\left[
\proj0_{\mathsf C}\otimes U_\theta
+\proj1_{\mathsf C}\otimes(I_{\mathsf A}\otimes(\pm Z_{\mathsf S}))
\right]
(U_h\otimes I_{\mathsf A\mathsf S}),
\]
where $\mathsf C$ and $\mathsf A$ are single-qubit ancillas, $U_h$ is a
unitary on $\mathsf C$, and
\[
U_\theta:=\sin\theta\,Z_{\mathsf A}\otimes Z_{\mathsf S}
+\cos\theta\,X_{\mathsf A}\otimes I_{\mathsf S},
\qquad
U_h\ket0_{\mathsf C}
=\frac1{\sqrt c}\ket0_{\mathsf C}
+\sqrt{\frac{2h}{c}}\ket1_{\mathsf C}.
\]
The two Pauli terms in $U_\theta$ anticommute, so $U_\theta$ is a Hermitian
unitary.  The bracketed operator is therefore a block-diagonal Hermitian
unitary, so $B_{\theta,\pm}$ is also a Hermitian unitary.  Its block on the
ancilla state $\ket0_{\mathsf C}\ket0_{\mathsf A}$ is
\begin{equation}
(\bra0_{\mathsf C}\bra0_{\mathsf A}\otimes I_{\mathsf S})B_{\theta,\pm}
(\ket0_{\mathsf C}\ket0_{\mathsf A}\otimes I_{\mathsf S})
=\frac{\sin\theta\pm2h}{c}\,Z_{\mathsf S}
=H_{\theta,\pm}.
\label{eq:shifted-block-encodings}
\end{equation}

Whenever $|\sin\theta|\ne2h$, both Hamiltonians are nonzero scalar
multiples of $Z$ and have unique ground states.  A positive coefficient
makes $\ket1$ the ground state, and a negative coefficient makes $\ket0$
the ground state.  The two ground states differ when
$|\sin\theta|<2h$; when $\sin\theta>2h$ they are both $\ket1$, and when
$\sin\theta<-2h$ they are both $\ket0$.

Use the same state-preparation oracle $U_I\ket0=\ket+$ for both families,
where $\ket+:=(\ket0+\ket1)/\sqrt2$.  Run the assumed preparation algorithm
independently on $B_{\theta,+}$ and $B_{\theta,-}$, and measure both output
system qubits in the computational basis.  Write $b_\pm$ for the success
flags and $z_\pm$ for the measurement outcomes.  Accept when both
runs succeed and the outcomes are $1,0$, respectively, and let
\[
p(\theta):=\Pr[b_+=b_-=1,\ z_+=1,\ z_-=0].
\]
If each run makes at most $T$ Hamiltonian queries, the combined circuit
makes at most $2T$ queries.  Every matrix entry of $B_{\theta,\pm}$ and its
inverse has trigonometric degree at most one in $\theta$, whereas all other
operations, including $U_I$ and $U_I^\dagger$, are independent of $\theta$.
Consequently, each final amplitude has degree at most $2T$, and taking
squared moduli and summing over the accepting outcomes shows that
$p(\theta)$ is a real trigonometric polynomial of degree at most $4T$.
This is the polynomial-method argument of~\cite[Claim~5.3]{MdW26}.
The oracle unitaries, and hence this polynomial, are defined for every
$\theta$, including values outside the spectral promise.

For $|\sin\theta|\ge3h$, the magnitude of each coefficient of $Z$ is at
least $h/c$, so the two Hamiltonians have spectral gaps at least $2h/c$.
Set
\[
h:=\frac{\delta}{1-2\delta},
\qquad
c=\frac1{1-2\delta},
\qquad
\frac hc=\delta.
\]
Their ground energies are then at most $-\delta$ and their excited energies
are at least $\delta$, so their gaps are at least $2\delta$ and they satisfy
\eqref{eq:lower-normalized-promise} with separator zero.  Moreover,
$\ket+$ has overlap $1/\sqrt2\ge\gamma$ with either ground state.  If
$\sin\theta\ge3h$, both ground states are $\ket1$, so acceptance requires
the minus run to produce an incorrect outcome.  The conditional
trace-distance guarantee implies
\[
p(\theta)\le\Pr[b_-=1,z_-=0]
=\Pr[b_-=1]\Pr[z_-=0\mid b_-=1]
\le\varepsilon.
\]
If $\sin\theta\le-3h$, both ground states are $\ket0$; acceptance then
requires the plus run to output $1$, so
$p(\theta)\le\Pr[b_+=1,z_+=1]\le\varepsilon$.  Therefore
\begin{equation}
0\le p(\theta)\le\varepsilon
\qquad\text{whenever }|\sin\theta|\ge3h.
\label{eq:acceptance-away-from-zero}
\end{equation}

At $\theta=0$, the two Hamiltonians are $H_{0,\pm}=\pm2\delta Z$, so they
also satisfy the promise, with spectral gaps $4\delta$ and ground states
$\ket1$ and $\ket0$, respectively.  By independence, the probability that
both runs succeed is at least $s_0^2$.  Conditioned on their success, both specified
measurement outcomes occur with probability at least $(1-\varepsilon)^2$.
Hence
\begin{equation}
p(0)\ge s_0^2(1-\varepsilon)^2,
\label{eq:acceptance-at-zero}
\end{equation}
which is at least $s_0^2/4$ because $\varepsilon<1/2$.

Thus $p(\theta)$ is at most $\varepsilon$ outside
\[
E:=\{\theta\in[-\pi,\pi):|\sin\theta|<3h\},
\]
but is at least $s_0^2/4$ at zero.  Since $\delta\le1/16$, we have
$h\le1/14$, $h\le8\delta/7$, and
\[
|E|=4\arcsin(3h)\le24h<28\delta,
\qquad
|E|\le4\arcsin(3/14)<\frac\pi2.
\]
Apply the growth inequality from Lemma~\ref{lem:trig-growth} to
$q(\theta):=p(\theta)/\varepsilon$, with degree at most $4T$ and an
exceptional set of length $|E|$.  The lower bound at $\theta=0$ from
\eqref{eq:acceptance-at-zero}, together with the bound outside $E$ in
\eqref{eq:acceptance-away-from-zero}, yields
\[
\frac{s_0^2(1-\varepsilon)^2}{\varepsilon}
\le\sup_\theta|q(\theta)|
\le\exp\!\bigl(4(4T)(28\delta)\bigr).
\]
Taking logarithms and using $(1-\varepsilon)^2\ge1/4$ gives
\begin{equation}
448T\delta\ge
\log\frac{s_0^2(1-\varepsilon)^2}{\varepsilon}
\ge\log\frac{s_0^2}{4\varepsilon}.
\label{eq:precision-explicit-lower}
\end{equation}
This proves the lemma.
\end{proof}

Combining this precision lower bound with Lemma~\ref{lem:overlap-lower}
gives the lower bound on $\mathbb E[q_H]$ in
Theorem~\ref{thm:expected-query-complexity}.  We first truncate an algorithm
with a random query count to one with a fixed query bound.

\begin{proof}[Proof of the {$\mathbb E[q_H]$} lower bound in Theorem~\ref{thm:expected-query-complexity}]
Set $\delta:=\Delta/(2\alpha)\le1/16$, and suppose an
algorithm outputs a state within trace distance $\varepsilon$ of
$\proj{\psi_0}$ with
$\mathbb E[q_H]\le T$ on every valid instance.  Stop it before its
$(\lfloor2T\rfloor+1)$st Hamiltonian query and report failure if it has not
yet terminated.  The truncated algorithm uses at most $2T$ queries, and
Markov's inequality gives success probability $p\ge1/2$.  If $\rho$ is the
original output and $\rho_{\rm succ}$ is the output conditioned on
termination within the budget, then
\begin{equation}
\Dtr(\rho_{\rm succ},\proj{\psi_0})
\le\sqrt{1-\bra{\psi_0}\rho_{\rm succ}\ket{\psi_0}}
\le\sqrt{\frac{1-\bra{\psi_0}\rho\ket{\psi_0}}p}
\le\sqrt{2\varepsilon}=: \eta.
\label{eq:expected-truncation-error}
\end{equation}
Since $\eta<9/20$, Lemma~\ref{lem:overlap-lower} with $s_0=1/2$ gives
$T=\Omega(1/(\gamma\delta))$.  For sufficiently small $\varepsilon$,
Lemma~\ref{lem:precision-lower} also applies and gives
$T=\Omega(\delta^{-1}\log(1/\varepsilon))$, because
$\log(1/\eta)=\tfrac12\log(1/\varepsilon)-\cO(1)$.
For the remaining $\varepsilon\le1/10$,
$\delta^{-1}\log(1/\varepsilon)=\cO((\gamma\delta)^{-1})$, so the
overlap lower bound suffices.  Taking the maximum of the
two bounds and substituting $\delta=\Delta/(2\alpha)$ yields
\[
T=\Omega\!\left(
\frac{\alpha}{\gamma\Delta}
+\frac{\alpha}{\Delta}\log\frac1\varepsilon
\right).
\]

\end{proof}

For the $Q_H$ lower bound in Theorem~\ref{thm:unconditional-output}, we use
the hard family of Somma and de Wolf for small $\gamma$ and
Lemma~\ref{lem:precision-lower} for constant $\gamma$.

\begin{proof}[Proof of the $Q_H$ lower bound in Theorem~\ref{thm:unconditional-output}]
Somma and de Wolf reduce their unique-ground-state hard family to block
encodings~\cite[Section~1.2]{SdW26}.  After negating the encoded
Hamiltonian, its eigenvalue on the unique ground state is $-\sin\theta$ and
every other eigenvalue is zero, where the nonzero eigenphase $\theta$ is
fixed across the family.  Scaling the Hamiltonian by $\alpha$ and taking
\[
\theta=\arcsin(\Delta/\alpha),
\qquad
\mu=-\frac\Delta2,
\]
gives ground energy $-\Delta$, excited energy zero, and the threshold
promise~\eqref{eq:separator-promise}.  Because
$\theta=\Theta(\Delta/\alpha)$ for $\Delta\le\alpha/8$, their lower bound
gives, for sufficiently small $\gamma$,
\[
Q_H=\Omega\!\left(
\frac{\alpha}{\gamma\Delta}\log\frac1\varepsilon
\right)
\]
in dimension $\cO(\gamma^{-2}\log^2(1/\varepsilon))$.
For $\gamma$ bounded below by an absolute constant,
Lemma~\ref{lem:precision-lower} with $s_0=1$ and
$\delta=\Delta/(2\alpha)$ gives
$Q_H=\Omega((\alpha/\Delta)\log(1/\varepsilon))$, which is the same bound
because $1/\gamma=\cO(1)$.
\end{proof}

\subsection{State-preparation queries}
\label{subsec:initial-query-lower}

We prove lower bounds on queries to $U_I$ when
$Q_H=o(\sqrt N/\delta)$, where the system dimension is $N+1$ and
$2\delta$ is the normalized gap.  Without calling $U_I$, one can prepare
a state within trace distance $\varepsilon$ of the ground state on success,
with constant success probability, using
$\cO((\sqrt N+\log(1/\varepsilon))/\delta)$ calls to $U_H$.  Prepare
the maximally entangled state
\[
\ket\Phi=\frac1{\sqrt{N+1}}\sum_{x=0}^N\ket x\ket x,
\qquad
\norm{(\proj{\psi_0}\otimes I)\ket\Phi}^2=\frac1{N+1},
\]
where the second register is a reference.  The reflection about the target
subspace is $R_g\otimes I$.  Tensoring the transducer of
Proposition~\ref{prop:ground-reflection} with the identity on the reference
gives this transduction action with constant transduction complexity.
Amplitude amplification from $\ket\Phi$, implemented
as in Section~\ref{sec:constant-precision-amplification}, raises the overlap with this subspace
to a constant using $\cO(\sqrt N/\delta)$ calls to $U_H$.  The reflection about
$\ket\Phi$ uses no input oracle.  One spectral filter with accuracy
$\Theta(\varepsilon)$ costs $\cO(\delta^{-1}\log(1/\varepsilon))$ further
calls to $U_H$.  On its accepting outcome, discarding the reference gives
the claimed approximation.  We now prove the query tradeoff of
Theorem~\ref{thm:initial-query-tradeoff}.

\begin{proof}[Proof of Theorem~\ref{thm:initial-query-tradeoff}]
We construct $N$ instances with ground state $\ket t$ and compare each
computation with one using fixed oracles independent of $t$.
Set $\delta=\Delta/(2\alpha)$ and $c_0=s_0(1-\varepsilon)$.
On the basis $\{\ket\bot,\ket1,\ldots,\ket N\}$, with $\ket\bot:=\ket0$, define
\[
P_t:=\proj t,
\qquad
\widetilde H_t:=\delta I-2\delta P_t,
\qquad
U_{H,t}:=\sqrt{1-\delta^2}\,X\otimes I
+\delta Z\otimes(I-2P_t).
\]
$\widetilde H_t\ket t=-\delta\ket t$, while
$\widetilde H_t\ket x=\delta\ket x$ for $x\ne t$.  Moreover,
\[
U_{H,t}^2=(1-\delta^2)I+\delta^2(I-2P_t)^2=I,
\qquad
(\bra0\otimes I)U_{H,t}(\ket0\otimes I)
=\delta(I-2P_t)=\widetilde H_t.
\]
The cross terms in $U_{H,t}^2$ vanish because $XZ+ZX=0$.  Thus $\ket t$
is the unique ground state, the gap is $2\delta$, and $U_{H,t}$ is an exact
$1$-block-encoding of $\widetilde H_t$.
Choose the guiding-state oracle to act as
\[
U_{I,t}\ket\bot=\sqrt{1-\gamma^2}\ket\bot+\gamma\ket t,
\qquad
U_{I,t}\ket t=-\gamma\ket\bot+\sqrt{1-\gamma^2}\ket t,
\]
and as the identity on the orthogonal subspace.  Rescaling to
$H_t=\alpha\widetilde H_t$ gives valid instances for the theorem.
A successful preparation followed by a computational-basis measurement
returns the label $t$ with joint probability at least $c_0$.

Introduce the reference oracles
\[
U_{I,0}=I,
\qquad
U_{H,0}=\sqrt{1-\delta^2}\,X\otimes I+\delta Z\otimes I.
\]
For every normalized state $\ket\Phi$,
\begin{align*}
\sum_{t=1}^N\norm{(U_{I,t}-U_{I,0})\ket\Phi}^2
&\le N\max_t\norm{U_{I,t}-I}^2
\le2N\gamma^2,\\
\sum_{t=1}^N\norm{(U_{H,t}-U_{H,0})\ket\Phi}^2
&=4\delta^2\sum_{t=1}^N\norm{(I\otimes P_t)\ket\Phi}^2
\le4\delta^2.
\end{align*}
For inverse queries, $U_{H,t}^\dagger=U_{H,t}$ and
$U_{I,t}^\dagger-I=(U_{I,t}-I)^\dagger$.  For a controlled call,
\[
\widehat U_{q,t}-\widehat U_{q,0}
=\proj1_{\mathsf Q}\otimes(U_{q,t}-U_{q,0}).
\]
Thus the inequalities for $U_{I,t}$ and $U_{H,t}$ also hold for inverse and
controlled queries.  Write $\ket{\Psi^t}$ and $\ket{\Psi^0}$ for the final
states on these instances and on the reference oracles.  Apply
Lemma~\ref{lem:query-hybrid} first to $U_{I,t}$ with $U_{H,t}$ fixed for
each $t$, then to $U_{H,t}$ with $U_I=I$ fixed.  The triangle inequality gives
\begin{equation}
D:=\left(\sum_{t=1}^N\norm{\ket{\Psi^t}-\ket{\Psi^0}}^2\right)^{1/2}
\le2\gamma\sqrt N\,Q_I+2\delta Q_H.
\label{eq:two-oracle-hybrid-upper}
\end{equation}

Let $p_t^0$ be the reference probability of success and outcome $t$.
Since $\sum_t p_t^0\le1$, at least $N/2$ labels satisfy
$p_t^0\le2/N\le c_0/2$.  For each such label, the actual joint probability
is at least $c_0$, so contractivity of trace distance gives
$\norm{\ket{\Psi^t}-\ket{\Psi^0}}\ge c_0/2$.  Consequently,
\[
2\gamma\sqrt N\,Q_I+2\delta Q_H
\ge D\ge\frac{c_0\sqrt N}{2\sqrt2},
\]
and hence
\[
\gamma Q_I+\frac{\delta Q_H}{\sqrt N}
\ge\frac{c_0}{4\sqrt2}.
\]
Substituting $\delta=\Delta/(2\alpha)$ proves the theorem.
\end{proof}

The tradeoff gives Corollary~\ref{cor:initial-query-lower} when its $Q_H$
term vanishes.

\begin{proof}[Proof of Corollary~\ref{cor:initial-query-lower}]
The Hamiltonian-query term in Theorem~\ref{thm:initial-query-tradeoff}
is $o(1)$, while $s_0(1-\varepsilon)\ge9s_0/10$.
Subtracting that term gives $Q_I=\Omega_{s_0}(1/\gamma)$.
\end{proof}

For the lower bound on $\mathbb E[q_I]$ in
Theorem~\ref{thm:expected-query-complexity}, we truncate both query counts
and apply the same tradeoff.

\begin{proof}[Proof of the {$\mathbb E[q_I]$} lower bound in Theorem~\ref{thm:expected-query-complexity}]
Let $B_H,B_I$ be the suprema of the expected query counts over
valid $(N+1)$-dimensional instances.  If $B_I$ is infinite, the claim is
immediate.  Stop before either count exceeds $\lfloor4B_H\rfloor$ or
$\lfloor4B_I\rfloor$ and report failure.  Markov's inequality and the union bound give
success probability at least $1/2$, and~\eqref{eq:expected-truncation-error}
gives conditional error at most $\sqrt{2\varepsilon}$.
Theorem~\ref{thm:initial-query-tradeoff} applied to this truncated algorithm yields
\[
4\gamma B_I+\frac{2\Delta B_H}{\alpha\sqrt N}
\ge\frac{1-\sqrt{2\varepsilon}}{8\sqrt2}.
\]
Since $B_H=o((\alpha/\Delta)\sqrt N)$, this proves
$B_I=\Omega(1/\gamma)$ for all sufficiently large $N$.
\end{proof}

For the worst-case query model, we also need a state-preparation-query lower
bound with its dependence on $\varepsilon$.
\begin{lemma}[Unconditional state-preparation lower bound]
\label{lem:unconditional-initial-query-lower}
Fix $0<\varepsilon\le1/10$.  For each $N$, let $0<\gamma\le1/\sqrt2$ and
$0<\delta<1$, possibly depending on $N$.  Suppose an algorithm outputs a
state $\rho$ satisfying $\Dtr(\rho,\proj{\psi_0})\le\varepsilon$ on every
$(N+1)$-dimensional instance satisfying~\eqref{eq:lower-normalized-promise}.
If $Q_H=o(\sqrt N/\delta)$ uniformly over these instances, then for all
sufficiently large $N$ some instance satisfies
\begin{equation}
Q_I=\Omega\!\left(\frac1\gamma\log\frac1\varepsilon\right).
\label{eq:unconditional-initial-query-lower}
\end{equation}
\end{lemma}

\begin{proof}
To obtain a polynomial whose degree depends only on $Q_I$, keep
$\widetilde H_t$ and $U_{H,t}$ from the proof of
Theorem~\ref{thm:initial-query-tradeoff} fixed and vary the guiding state as
\[
\ket{\psi_{t,\theta}}=\cos\theta\ket\bot+\sin\theta\ket t,
\qquad \theta\in[-\pi,\pi).
\]
Its preparation oracle is
\[
U_{I,t}(\theta)
=I+(\cos\theta-1)(\proj\bot+P_t)
+\sin\theta(\ket t\!\bra\bot-\ket\bot\!\bra t).
\]
This is a valid instance whenever $|\sin\theta|\ge\gamma$.
Let $p_t(\theta)$ be the probability of outcome $t$ when the algorithm's
output is measured in the computational basis.  Each $U_I$ query has
trigonometric degree one in $\theta$, while $U_H$ is independent of
$\theta$, so $p_t$ is a real trigonometric polynomial of degree at most
$2Q_I$.

At $\theta=0$, all the preparation oracles are the identity, so only
$U_{H,t}$ depends on $t$.  Comparing the computations using $U_{H,t}$
with the computation using $U_{H,0}$, Lemma~\ref{lem:query-hybrid} and
Cauchy--Schwarz give
\begin{align*}
\frac1N\sum_{t=1}^N p_t(0)
&\le\frac1N+\frac1N\sum_{t=1}^N
\norm{\ket{\Psi^t(0)}-\ket{\Psi^0}}\\
&\le\frac1N+\frac{2\delta Q_H}{\sqrt N}\le\frac12
\end{align*}
for all sufficiently large $N$.  Fix a label $t$ with $p_t(0)\le1/2$.
For every $|\sin\theta|\ge\gamma$, the output guarantee instead gives
$p_t(\theta)\ge1-\varepsilon$.

Set $q(\theta)=(1-p_t(\theta))/\varepsilon$.
Then $q(0)\ge1/(2\varepsilon)$ and $0\le q(\theta)\le1$ outside a set of
measure $s=4\arcsin\gamma\le\pi$.  Lemma~\ref{lem:trig-growth} therefore gives
\[
\frac1{2\varepsilon}\le q(0)\le\sup_\theta|q(\theta)|
\le\exp\!\left(32Q_I\arcsin\gamma\right).
\]
Taking logarithms and using $\arcsin\gamma=\Theta(\gamma)$ proves the claim.
\end{proof}

Rescaling the Hamiltonians in Lemma~\ref{lem:unconditional-initial-query-lower}
gives the $Q_I$ lower bound in Theorem~\ref{thm:unconditional-output}.

\begin{proof}[Proof of the $Q_I$ lower bound in Theorem~\ref{thm:unconditional-output}]
Set $\delta=\Delta/(2\alpha)$ and apply
Lemma~\ref{lem:unconditional-initial-query-lower}.  Rescaling its
Hamiltonians as in~\eqref{eq:lower-bound-rescaling} leaves both query counts
unchanged and gives the claimed bound.
\end{proof}

\section*{Statement on AI use}

The authors proposed applying the purifier of Belovs and Jeffery~\cite{BJ26} and the transducer composition results of Belovs, Jeffery, and Yolcu~\cite{BJY24} to ground-state preparation with block-encoding access. Large language models were used to develop the detailed constructions and proofs. The authors reviewed and verified the arguments, substantially rewrote the manuscript, and take full responsibility for its claims, proofs, and citations.

\appendix
\section{Proof of the overlap lower bound with controlled queries}
\label{app:overlap-lower}

\begin{proof}[Proof of Lemma~\ref{lem:overlap-lower}]
Set $N=\lfloor\gamma^{-2}\rfloor\ge2$ and $P_t=\proj t$ on the basis
$\{\ket\bot,\ket1,\ldots,\ket N\}$, where $\ket\bot=\ket0$.
For $1\le t\le N$, take
\[
\widetilde H_t=\delta(I-2P_t),\qquad
U_{H,t}=\sqrt{1-\delta^2}\,X\otimes I
          +\delta Z\otimes(I-2P_t).
\]
Anticommutation of $X$ and $Z$ gives $U_{H,t}^2=I$, and the $\ket0$
block of $U_{H,t}$ is $\widetilde H_t$.  Thus $\ket t$ is the unique
ground state and \eqref{eq:lower-normalized-promise} holds.  Use the
same $U_I$ for every $t$, with
$U_I\ket\bot=N^{-1/2}\sum_{r=1}^N\ket r$; its overlap with each
ground state is at least $\gamma$.

Take the reference oracle
$U_{H,0}=\sqrt{1-\delta^2}\,X\otimes I+\delta Z\otimes I$.
For every normalized $\ket\Phi$,
\[
\sum_{t=1}^N\norm{(U_{H,t}-U_{H,0})\ket\Phi}^2
=4\delta^2\sum_{t=1}^N\norm{(Z\otimes P_t)\ket\Phi}^2
\le4\delta^2.
\]
The bound also holds for controlled calls; inverse calls are identical
because the oracles are Hermitian.  As $U_I$ is independent of $t$,
Lemma~\ref{lem:query-hybrid} gives
\[
\left(\sum_{t=1}^N\norm{\ket{\Psi^t}-\ket{\Psi^0}}^2\right)^{1/2}
\le2\delta T.
\]
Let $p_t$ and $p_t^0$ be the probabilities of success and system outcome
$t$ with $U_{H,t}$ and $U_{H,0}$, respectively, and set
$c=s_0(1-\varepsilon)$.  Then $p_t\ge c$ and $\sum_t p_t^0\le1$.
Since $|p_t-p_t^0|\le\norm{\ket{\Psi^t}-\ket{\Psi^0}}$,
Cauchy--Schwarz gives
\[
Nc-1
\le\sum_{t=1}^N|p_t-p_t^0|
\le2\delta T\sqrt N.
\]
If $N\ge2/c$, then $N\ge1/(2\gamma^2)$ gives
$T=\Omega_{s_0}(1/(\gamma\delta))$.

If $N<2/c$, then $\gamma=\Omega_{s_0}(1)$.  Comparing $t=1,2$ with
$M=\proj1_{\mathsf B}\otimes(P_1-P_2)$ and
$\norm{U_{H,1}-U_{H,2}}=2\delta$ gives
\[
2s_0(1-2\varepsilon)
\le\bigl|\bra{\Psi^1}M\ket{\Psi^1}
       -\bra{\Psi^2}M\ket{\Psi^2}\bigr|
\le4\delta T.
\]
This again yields $T=\Omega_{s_0}(1/(\gamma\delta))$.
\end{proof}

\newpage

\bibliographystyle{alphaurl}
\bibliography{ref}

\end{document}